\documentclass{article}

\usepackage{amsmath}

\usepackage{amsthm} % for proof environment
\usepackage{amssymb }
\usepackage{bbm}
\usepackage{xspace}
\usepackage{enumitem}
\usepackage{verbatim}
\usepackage{ wasysym }
\usepackage{hyperref}
\usepackage{caption}
\usepackage{tikz}
\usetikzlibrary{arrows.meta}
\usetikzlibrary{automata,positioning}

\newtheorem{theorem}{Theorem}[section]

\newtheorem{corollary}[theorem]{Corollary}
\newtheorem{lemma}[theorem]{Lemma}

\newtheorem{proposition}[theorem]{Proposition}
\newtheorem{remark}{Remark}[section]

\newcommand\stam[1]{}
\newcommand{\buchi}{B\"uchi\xspace}
\newcommand{\A}{{\cal A}}
\newcommand{\B}{{\cal B}}
\newcommand{\C}{{\cal C}}
\newcommand{\D}{{\cal D}}
\renewcommand{\S}{{\cal S}}

\newcommand{\N}{{\cal N}}

\newcommand{\zug}[1]{\langle #1 \rangle}

\newcommand{\im}[1]{\infty #1}
\newcommand{\comp}[1]{\overline{#1}}

\newcommand{\badernew}[1]{\textcolor{blue}{#1}}

\title{On Semantically-Deterministic Automata\thanks{A preliminary version appears in the Proceedings of the 50th International Colloquium on Automata, Languages, and Programming, 2023. Research supported by the Israel Science Foundation, Grant 2357/19, and the European Research Council, Advanced Grant ADVANSYNT.}}

\author{Bader Abu Radi and Orna Kupferman\\ School of Computer Science and Engineering\\
The Hebrew University, Jerusalem, Israel}

\begin{document}

\maketitle

%TODO mandatory: add short abstract of the document
\begin{abstract}
%Nondeterminism is a fundamental notion in Theoretical Computer Science.  
A nondeterministic automaton is {\em semantically deterministic\/} (SD) if different nondeterministic choices in the automaton lead to equivalent states. Semantic determinism is interesting as it is a natural relaxation of determinism, and as some applications of deterministic automata in formal methods can actually use automata with some level of nondeterminism, tightly related to semantic determinism.

In the context of finite words, semantic determinism coincides with determinism, in the sense that every pruning of an SD automaton to a deterministic one results in an equivalent automaton. We study SD automata on infinite words, focusing on B\"uchi, co-\buchi, and weak automata. We show that there, while semantic determinism does not increase the expressive power, the combinatorial and computational properties of SD automata are very different from these of deterministic automata.
In particular, SD \buchi and co-\buchi automata are exponentially more succinct than deterministic ones (in fact, also exponentially more succinct than history-deterministic automata), their complementation involves an exponential blow up, and decision procedures for them like universality and minimization are PSPACE-complete. For weak automata, we show that while an SD weak automaton need not be pruned to an equivalent deterministic one, it can be determinized to an equivalent deterministic weak automaton with the same state space, implying also efficient complementation and decision procedures for SD weak automata.
\end{abstract}

\section{Introduction}
\label{intro}

{\em Automata\/} are among the most studied computation models in theoretical computer science. Their simple structure has made them a basic formalism for the study of fundamental notions, such as \emph{determinism} and {\em nondeterminism} \cite{RS59}.  
While a deterministic computing machine examines a single action at each step of its computation, nondeterministic machines are allowed to examine several possible actions simultaneously. 
Understanding the power of nondeterminism is at the core of open fundamental questions in theoretical computer science (most notably, the P vs. NP problem).

A prime application of {\em automata on infinite words} is specification, verification, and synthesis of nonterminating systems. The automata-theoretic approach reduces questions about systems and their specifications to questions about automata \cite{Kur94,VW94}, and is at the heart of many algorithms and tools.
A run of an automaton on infinite words is an infinite sequence of states, and acceptance is determined with respect to the set of states that the run visits infinitely often. For example, in {\em B\"uchi\/} automata, some of the states are designated as accepting states, and a run is accepting iff it visits states from the set $\alpha$ of accepting states infinitely often \cite{Buc62}.
Dually, in {\em co-B\"uchi\/} automata, a run is accepting if it visits the set $\alpha$ only finitely often.
Then, {\em weak\/} automata are a special case of both \buchi and co-\buchi automata in which every strongly connected component in the graph induced by the automaton is either contained in $\alpha$ or is disjoint from $\alpha$. We use DBW and NBW to denote determinisitic and nondeterministic \buchi word automata, respectively, and similarly for D/NCW, D/NWW, and D/NFW, for co-\buchi, weak, and automata on finite words, respectively.

For automata on infinite words, nondeterminism not only leads to exponential succinctness, but may also increase the expressive power. This is the case, for example, in \buchi and weak automata, thus NBWs are strictly more expressive than DBWs~\cite{Lan69}, and NWWs are strictly more expressive than DWWs~\cite{BJW01}. On the other hand, NCWs are as expressive as DCWs~\cite{MH84}, and in fact, also as NWWs \cite{Kur87}. In some applications of the automata-theoretic approach, such as model checking, algorithms can be based on nondeterministic automata, whereas in other applications, such as synthesis and reasoning about probabilistic systems, they cannot. There, the advantages of nondeterminism are lost, and algorithms involve a complicated determinization construction~\cite{Saf88} or acrobatics for circumventing determinization~\cite{KV05c,Kup06a}.

In a deterministic automaton, the transition function maps each state and letter to a single successor state.
In recent years there is growing research on weaker types of determinism. This includes, for example, {\em unambiguous automata}, which may have many runs on each word, yet only one accepting run \cite{CM03,CQS21}, automata that are {\em deterministic in the limit}, where each accepting run should eventually reach a deterministic sub-automaton \cite{SEJK16}, and automata that are {\em determinizable by pruning} (DBP), thus embody an equivalent deterministic automaton~\cite{AKL10}.

In terms of applications, some weaker types of determinism have been defined and studied with specific applications in mind. Most notable are {\em history-deterministic} automata (HD), which can resolve their nondeterministic choices based on the history of the run \cite{HP06,BL23,Kup23}\footnote{The notion used in \cite{HP06} is {\em good for games} (GFG) automata, as they address the difficulty of playing games on top of a nondeterministic automaton. As it turns out, the property of being good for games varies in different settings and HD is good for applications beyond games. Therefore, we use the term {\em history determinism}, introduced by Colcombet in the setting of quantitative automata with cost functions \cite{Col09}.}, and can therefore replace deterministic automata in algorithms for synthesis and control, and {\em good-for-MDPs} automata (GFM), whose product with Markov decision processes maintains the probability of acceptance, and can therefore replace deterministic automata when reasoning about stochastic behaviors \cite{HPSSTWW20,STZ22}.

The different levels of determinism induce classes of automata that differ in their succinctness and in the complexity of operations and decision problems on them. Also, some classes are subclasses of others. For example, it follows quite easily from the definitions that every automaton that is deterministic in the limit is GFM, and every automaton that is DBP is HD.

In this paper we study the class of {\em semantically deterministic\/} automata (SD). An automaton $\A$ is SD if its nondeterministic choices lead to equivalent states. Formally, if $\A=\zug{\Sigma,Q,q_0,\delta,\alpha}$, with a transition function $\delta:Q \times  \Sigma \rightarrow 2^Q$, then $\A$ is SD if for every state $q \in Q$, letter $\sigma \in \Sigma$, and states $q_1,q_2 \in \delta(q,\sigma)$, the set of words accepted from $\A$ with initial state $q_1$ is equal to  the set of words accepted from $\A$ with initial state $q_2$. Since all nondeterministic choices lead to equivalent states, one may be tempted to think that SD automata are DBP or at least have similar properties to deterministic automata. This is indeed the case for SD-NFWs, namely when one considers automata on finite words. There, it is not hard to prove that any pruning of an SD-NFW to a DFW results in an equivalent DFW.  Thus, SD-NFWs are not more succinct than DFWs, and operations on them are not more complex than operations on DFWs.

Once, however, one considers automata on infinite words, the simplicity of SD automata is lost.
In order to understand the picture in the setting of infinite words, let us elaborate some more on HD automata, which are strongly related to SD automata.
Formally, a nondeterministic automaton $\A$ is HD if there is a strategy $g$ that maps each finite word $u \in \Sigma^*$ to a transition to be taken after $u$ is read, and following $g$ results in accepting all the words in the language of $\A$. %Note that a state $q$ of $\A$ may be reachable via different words, and $g$ may suggest different transitions from $q$ after different words are read. Still, $g$ depends only on the past, namely on the word read so far.
Obviously, every DBP automaton is HD -- the strategy $g$ can suggest the same transition in all its visits in a state. On the other hand, while HD-NWWs are always DBP~\cite{KSV06,Mor03}, this is not the case for HD-NBWs and HD-NCWs~\cite{BKKS13}. There, the HD strategy may need to suggest different transitions in visits (with different histories) to the same state. 
%Essentially, this follows from the fact that if a state $q$ has a $\sigma$-successor $q'$ whose language is strictly contained in the language of another $\sigma$-successors of $q$, for some letter $\sigma$, then no HD strategy takes the transition from $q$ to $q'$, and so this transition can be pruned. 

It is easy to see that a strategy $g$ as above cannot choose a transition to states whose language is strictly contained in the language of states that are reachable by other transitions. Thus, all HD automata can be pruned in polynomial time to SD automata~\cite{KS15,BK18}. The other direction, however, does not hold: it is shown in \cite{AKL21} that an SD-NWW need not be HD (hence, SD-NBWs and SD-NCWs need not be HD too). Moreover, while all all HD-NWWs are DBP, this is not the case for all SD-NWWs.  

In this work we study the succinctness of  SD automata with respect to deterministic ones, as well as the complexity of operations on them and decision problems about them. Our goal is to understand the difference between determinism and semantic determinism, and to understand how this difference varies among different acceptance conditions. Our study is further motivated by the applications of automata with different levels of nondeterminism in algorithms for synthesis and for reasoning in a stochastic setting. In particular, beyond the connection to HD automata discussed above, as runs of an SD-NBW on words in the language are accepting with probability 1, all SD-NBWs are GFM \cite{AKL21}.

We study semantic determinism for \buchi, co--\buchi, and weak automata. We consider automata with both {\em state-based\/} acceptance conditions, as defined above, and {\em transition-based\/} acceptance conditions. In the latter, the acceptance condition is given by a subset $\alpha$ of transitions, and a run is required to traverse transitions in $\alpha$ infinitely often (in B\"uchi automata, termed tD/tNBW) or finitely often (in co-B\"uchi automata, termed tD/tNCW). As it turns out, especially in the context of HD automata, automata with transition-based acceptance conditions may differ in their properties from automata with traditional state-based acceptance conditions. For example, while HD-tNCWs can be minimized in PTIME \cite{AK19}, minimization of HD-NCWs is NP-complete \cite{Sch20}. In addition, there is recently growing use of transition-based automata in practical applications, with evidences they offer a simpler translation of LTL formulas to  automata and enable simpler constructions and decision procedures  \cite{GO01,GL02,DLFMRX16,SEJK16,LKH17}. Our results for all types of acceptance conditions are summarized in Table~\ref{t1} below, where we also compare them with known results about deterministic and HD automata.

\begin{table}[h]
	\footnotesize
	\begin{center}
		\begin{tabular}{|c||c|c|c|}
			& Deterministic & HD   & SD  \\ \hline \hline
			Succinctness & $n$ (B,C,W) & $n^2$ (B), $2^n$ (C), $n$ (W) & $2^n$ (B,C), $n$ (W)\\
			& & \cite{KS15,KSV06,Mor03} & Theorems \ref{SD-tNBW succinct thm},~\ref{SD-tNCW succinct thm},~\ref{sd nww det} \\ \hline
			Complementation & $n$ (B,C,W) & $n$ (B), $2^n$ (C), $n$ (W) & $2^n$ (B,C), $n$ (W)\\
			& \cite{Kup18} & \cite{KS15} & Theorems \ref{SD-tNBW comp thm},~\ref{SD-tNCW comp thm},~\ref{sd nww comp}\\ \hline
			Universality & NL (B,C,W) & \ \ \ P (B,C,W) \ \ \ & \ \ PSPACE (B,C), P (W) \ \  \\
			& \cite{Kup18} & \cite{HKR02,KS15} & Theorems \ref{universality buchi thm},~\ref{universality SD-tNCWs thm},~\ref{sd nww dp} \\ \hline
			Minimization & NP (B,C), P (W)  & NP (B,C), P (W)  & PSPACE (B,C), P (W)  \\
			(state based)& \cite{Sch10,Lod01}  & \cite{Sch20,KSV06,Lod01} & Theorems \ref{min buchi is hard thm},~\ref{min co is hard thm},~\ref{sd nww dp} \\ \hline
			Minimization & open (B,C), P (W) & open (B), P (C,W) & PSPACE (B,C), P (W)  \\
			(transition based) & \cite{Lod01}   & \cite{AK19,KSV06,Lod01} & Theorems \ref{min buchi is hard thm},~\ref{min co is hard thm},~\ref{sd nww dp} \\ \hline
		\end{tabular}
		\normalsize
		\vspace{2mm}
		\caption{Succinctness (determinization blow-up), complementation (blow-up in going from an  automaton to a complementing one), universality (deciding whether an automaton accepts all words), and minimization (deciding whether an equivalent automaton of a given size exists, and the described results apply also for the case the given automaton is deterministic). All blow-ups are tights, except for HD-NBW determinization, where the quadratic bound has no matching lower bound; all NL, NP, and PSPACE bounds are complete.}
		\label{t1}
	\end{center}
\end{table}

Let us highlight the results we find the most interesting and surprising.
While all three types of SD automata are not DBP, we are able to determinize SD-NWWs in polynomial time, and end up with a DWW whose state space is a subset of the state space of the original SD-NWW. Essentially, rather than pruning transitions, the construction redirects transitions to equivalent states in deep strongly connected components of the SD-NWW, which we prove to result in an equivalent deterministic DWW\footnote{Our result implies that checking language-equivalence  between states in an SD-NWW can be done in PTIME, as the check can be performed on the equivalent DWWs. We cannot, however, use this complexity result in our algorithm, as this involves a circular dependency. Consequently, our construction of the equivalent DWWs involves a language-approximation argument.}.
This suggests that despite the ``not DBP anomaly" of SD-NWWs, they are very similar in their properties to DWWs. On the other hand, except for their expressive power, SD-NBWs and SD-NCWs are not at all similar to DBWs and DCWs: while HD-NBWs are only quadratically more succinct than DBWs, succinctness jumps to exponential for SD-NBWs. This also shows that SD-NBWs may be exponentially more succinct than HD-NBWs, and we show this succinctness gap also for co-\buchi automata, where exponential succinctness with respect to DCWs holds already in the HD level.  

The succinctness results are carried over to the blow-up needed for complementation, and to the complexity of decision procedures. Note that this is not the case for HD automata. There, for example, complementation of HD-NBWs results in an automaton with the same number of states \cite{KS15}. Moreover, even though HD-NCWs are exponentially more succinct than DCWs, language-containment for HD-NCWs can be solved in PTIME \cite{HKR02,KS15}, and HD-tNCWs can be minimized in PTIME \cite{AK19}.
For SD automata, we show that complementation involves an exponential blow-up. Also, universality and minimization of SD \buchi and co-\buchi automata, with either state-based or transition-based acceptance conditions, is PSPACE-complete, as it is for NBWs and NCWs.
We also study the {\em D-to-SD minimization problem}, where we are given a deterministic automaton $\A$ and a bound $k \geq 1$, and need to decide whether  %the problem of deciding whether 
$\A$ has an equivalent SD automaton with at most $k$ states. Thus, the given automaton is deterministic, rather than SD. By \cite{JR93}, the D-to-N minimization problem for automata on finite words (that is,  given a DFW, minimize it to an NFW) is PSPACE-complete. It is easy to see that the D-to-SD minimization problem for automata on finite words can be solved in PTIME. We show that while this is the case also for weak automata, D-to-SD minimization for B\"uchi and co-B\"uchi automata is PSPACE-complete.

Our results show that in terms of combinatorial and computational properties, semantic determinism is very similar to determinism in weak automata, whereas for B\"uchi and co-B\"uchi automata, semantic determinism is very similar to nondeterminism. The results are interesting also in comparison with history determinism, whose affect on B\"uchi and co-B\"uchi automata is different.

\section{Preliminaries}
\label{prelim}

\subsection{Languages and Automata}
For a finite nonempty alphabet $\Sigma$, an infinite {\em word\/} $w = \sigma_1 \cdot \sigma_2 \cdots \in \Sigma^\omega$ is an infinite sequence of letters from $\Sigma$.
A {\em language\/} $L\subseteq \Sigma^\omega$ is a set of words.
%
%orna1 got rid of the single place (gfg definition) we use it
%For an index $i \geq 0$, we use $w[1, i]$ to denote the (possibly empty) prefix $\sigma_1\cdot \sigma_2 \cdots  \sigma_i$ of $w$. 
For $1\leq i \leq  j$, we use $w[i,j]$ to denote the infix $\sigma_i \cdot \sigma_{i+1} \cdots \sigma_j$ of $w$, use $w[i]$ to denote the letter $\sigma_i$, and use  $w[i,\infty]$ to denote the infinite suffix $\sigma_i \cdot \sigma_{i+1} \cdots$ of $w$.
We also consider languages $R\subseteq  \Sigma^*$ of finite words, denote the empty word by $\epsilon$, and denote the set of nonempty words over $\Sigma$ by $\Sigma^+$; thus $\Sigma^+ = \Sigma^*\setminus \{\epsilon\}$.
For a set $S$, we denote its complement by $\comp{S}$. In particular,  for languages $R \subseteq \Sigma^*$ and $L\subseteq \Sigma^\omega$, we have $\overline{R} = \Sigma^*\setminus R$ and $\overline{L} = \Sigma^\omega  \setminus L$.

A \emph{nondeterministic automaton}  
is a tuple $\A = \langle \Sigma, Q, Q_0, \delta, \alpha  \rangle$,  where $\Sigma$ is an alphabet, $Q$ is a finite set of \emph{states}, $Q_0$ is a set of \emph{initial states}, $\delta: Q\times \Sigma \to 2^Q \setminus\emptyset$ is a \emph{transition function}, and $\alpha$ is an \emph{acceptance condition}, to be defined below. 
For states $q$ and $s$ and a letter $\sigma \in \Sigma$, we say that $s$ is a\emph{ $\sigma$-successor} of $q$ if $s \in \delta(q,\sigma)$.  
Note that the transition function of $\A$ is total, thus for all states $q\in Q$ and letters $\sigma\in \Sigma$, $q$ has at least one $\sigma$-successor.
%If $|\delta(q, \sigma)| = 1$ for every state $q\in Q$ and letter $\sigma \in \Sigma$, then $\A$ is \emph{deterministic}.
If $|Q_0| = 1$ and every state $q$ has a single $\sigma$-successor, for all letters $\sigma$, then $\A$ is \emph{deterministic}.
The transition function $\delta$ can be viewed as a transition relation $\Delta \subseteq Q\times \Sigma \times Q$, where for every two states $q,s\in Q$ and letter $\sigma\in \Sigma$, we have that $\langle q, \sigma, s \rangle \in \Delta$ iff $s\in \delta(q, \sigma)$.
We define the \emph{size} of $\A$, denoted $|\A|$, as its number of states, thus, $|\A| = |Q|$. 

Given an input word $w = \sigma_1 \cdot \sigma_2 \cdots$, a \emph{run}  of $\A$ on $w$ is a sequence of states $r = r_0,r_1,r_2,\ldots$, such that $r_0 \in Q_0$, and for all $i \geq 0$, we have that $r_{i+1} \in \delta(r_i, \sigma_{i+1})$, i.e, the run starts in some initial state and proceeds according to the transition function. If the word in the input is infinite, then so is the run.
We sometimes view the run $r = r_0,r_1,r_2,\ldots$ on $w = \sigma_1 \cdot \sigma_2 \cdots$ as a sequence of successive transitions $\zug{r_0,\sigma_1,r_1}, \zug{r_1,\sigma_2,r_2},\ldots$. %
Note that a deterministic automaton has a single run on an input $w$.
We sometimes extend $\delta$ to sets of states and finite words. Then, $\delta: 2^Q\times \Sigma^* \to 2^Q$ is such that for every $S \in 2^Q$, finite word $u\in \Sigma^*$, and letter $\sigma\in \Sigma$, we have that $\delta(S, \epsilon) = S$, $\delta(S, \sigma) = \bigcup_{s\in S}\delta(s, \sigma)$, and $\delta(S, u \cdot \sigma) = \delta(\delta(S, u), \sigma)$. Thus, $\delta(S, u)$ is the set of states that $\A$ may reach when it reads $u$ from some state in $S$.

The acceptance condition $\alpha$ determines which runs are ``good''. For automata on finite words, $\alpha \subseteq Q$, and a run is {\em accepting} if it ends in a state in $\alpha$.  
For automata on infinite words, we consider {\em state-based} and {\em transition-based} acceptance conditions.
Let us start with {\em state-based} conditions. Here, $\alpha \subseteq Q$, and we use the terms {\em $\alpha$-states\/} and  {\em $\overline{\alpha}$-states\/} to refer to states in $\alpha$ and in $Q \setminus \alpha$, respectively. For a run $r \in Q^\omega$, let ${\it sinf}(r)\subseteq Q$ be the set of states that $r$ visits infinitely often. Thus, ${\it sinf}(r) = \{q: \text{$q=r_i$ for infinitely many $i$'s}\}$. In {\em \buchi} automata, $r$ is \emph{accepting} iff ${\it sinf}(r)\cap \alpha \neq\emptyset$, thus if $r$ visits states in $\alpha$ infinitely often. Dually, in {\em co-\buchi} automata, $r$ is \emph{accepting} iff ${\it sinf}(r)\cap \alpha = \emptyset$, thus if $r$ visits states in $\alpha$ only finitely often.

We proceed to {\em transition-based} conditions. There, $\alpha\subseteq \Delta$ and acceptance depends on the set of transitions that are traversed infinitely often during the run. We use the terms {\em $\alpha$-transitions\/} and  {\em $\bar{\alpha}$-transitions\/} to refer to  transitions in $\alpha$ and in $\Delta \setminus \alpha$, respectively. For a run $r \in \Delta^\omega$, we define ${\it tinf}(r) =  \{  \langle q, \sigma, s\rangle \in \Delta: q = r_i, \sigma = \sigma_{i+1}, \text{ and } s = r_{i+1}, \text{ for infinitely many $i$'s}   \}$. As expected, in transition-based \buchi automata, $r$ is accepting iff ${\it tinf}(r) \cap \alpha \neq \emptyset$, and in transition-based co-\buchi automata, $r$ is accepting iff ${\it tinf}(r) \cap \alpha = \emptyset$.

Consider an automaton $\A = \langle \Sigma, Q, Q_0, \delta, \alpha  \rangle$.
In all automata classes, a run of $\A$ that is not accepting is \emph{rejecting}.  A word $w$ is accepted by an automaton $\A$ if there is an accepting run of $\A$ on $w$. The language of $\A$, denoted $L(\A)$, is the set of words that $\A$ accepts.

Consider a directed graph $G = \langle V, E\rangle$. A \emph{strongly connected set\/} in $G$ (SCS, for short) is a set $C\subseteq V$ such that for every two vertices $v, v'\in C$, there is a path from $v$ to $v'$. An SCS is \emph{maximal} if for every non-empty set $C'\subseteq V\setminus C$, it holds that $C\cup C'$ is not an SCS. The \emph{maximal strongly connected sets} are also termed \emph{strongly connected components} (SCCs, for short). 
%The \emph{SCC-DAG of $G$} is the directed acyclic graph defined over the SCCs of $G$, where there is an edge from a SCC $C$ to another SCC $C'$ iff there are two vertices $v\in C$ and $v'\in C'$ with $\langle v, v'\rangle\in E$.
 %A SCC is \emph{ergodic} iff it has no outgoing edges in the SCC graph. The SCC graph of $G$ can be computed in linear time by standard SCC algorithms \cite{Tar72}. 
An automaton $\A = \langle \Sigma, Q, Q_0, \delta, \alpha\rangle$ induces a directed graph $G_{\A} = \langle Q, E\rangle$, where $\langle q, q'\rangle\in E$ iff there is a letter $\sigma \in \Sigma$ such that $\langle q, \sigma, q'\rangle \in \Delta$. The SCSs and SCCs of $\A$ are those of $G_{\A}$. 

An automaton $\A = \zug{\Sigma, Q, Q_0, \delta, \alpha}$ %on infinite words 
with a state-based acceptance condition $\alpha \subseteq Q$ is \emph{weak}~\cite{MSS88} if for each SCC $C$ of $\A$, either $C\subseteq \alpha$ in which case $C$ is \emph{accepting}, or $C\cap \alpha = \emptyset$ in which case $C$ is \emph{rejecting}. We view $\A$ as a B\"uchi automaton, yet note that a weak automaton can be viewed as both a \buchi and a co-\buchi automaton. Indeed, a run of $\A$ visits $\alpha$ infinitely often iff it gets trapped in an SCC that is contained in $\alpha$ iff it visits states in $Q\setminus \alpha$ only finitely often. Note also that when $\A$ uses a transition-based acceptance condition, we can ignore the membership in $\alpha$ of transitions between SCCs (indeed, such transitions are traversed only finitely often), and say that $\A$ is weak if the transitions in each SCC are all in $\alpha$ or all disjoint from $\alpha$. 

Consider two automata $\A_1$ and $\A_2$. We say that $\A_1$ is {\em contained in\/} $\A_2$ if $L(\A_1) \subseteq L(\A_2)$. Then, $\A_1$ and $\A_2$ are {\em equivalent\/} if $L(\A_1) = L(\A_2)$, and $\A_1$ is {\em universal\/} if $L(\A_1) = \Sigma^\omega$ (or $L(\A_1) = \Sigma^*$, in case it runs on finite words). Finally, $\A_1$ is \emph{minimal} (with respect to a class of automata, say state-based deterministic \buchi automata) if for all automata $\A_2$ equivalent to $\A_1$, we have that $|\A_1| \leq |\A_2|$.

\subsection{SD and HD Automata}

Consider an automaton $\A = \langle \Sigma, Q, Q_0, \delta, \alpha  \rangle$.
For a state $q\in Q$ of $\A$, we define $\A^q = \langle \Sigma, Q, \{q\}, \delta, \alpha \rangle$, as the automaton obtained from $\A$ by setting the set of initial states to be $\{q\}$.
We say that two states $q,s\in Q$ are \emph{equivalent}, denoted $q \sim_{\A} s$, if $L(\A^q) = L(\A^s)$. %When $\A$ is clear from the context, we omit it from the notations, thus write $L(q)$, etc. Note that the $\sim_{\A}$ relation is transitive.
We say that {\em $q$ is reachable} if there is a finite word $x\in \Sigma^*$ with $q\in \delta(Q_0, x)$, and say that {\em $q$ is reachable  from $s$} if $q$ is reachable in $\A^s$.  
We say that $\A$ is \emph{semantically deterministic} (SD, for short) if different nondeterministic choices lead to equivalent states. Thus, all initial states are equivalent, and for every state $q\in Q$ and letter $\sigma \in \Sigma$, all the $\sigma$-successors of $q$ are equivalent. Formally, if $s,s' \in \delta(q, \sigma)$, then $s \sim_{\A} s'$. 
The following proposition, termed the \emph{SDness property}, follows immediately from the definitions and implies that in SD automata, for all finite words $x$, all the states in $\delta(Q_0, x)$ are equivalent. Intuitively, it means that a run of an SD automaton can take also bad nondeterministic choices, as long as it does so only finitely many times. 

\begin{proposition}\label{pruned-corollary}
	Consider an SD automaton $\A =\langle \Sigma, Q, Q_0, \delta, \alpha  \rangle$, and states $q,s \in Q$. If $q \sim_{\A} s$, then for every $\sigma\in \Sigma$, $q' \in \delta(q,\sigma)$, and $s' \in \delta(s,\sigma)$, we have that $q' \sim_{\A} s'$. 
\end{proposition}

A nondeterministic automaton $\A$ is \emph{history deterministic} (\emph{HD}, for short) if its nondeterminism can be resolved based on the past, thus on the prefix of the input word read so far. Formally, $\A$ is \emph{HD} if there exists a {\em strategy\/} $f:\Sigma^* \to Q$ such that the following hold: 
\begin{enumerate}
	\item 
	The strategy $f$ is consistent with the transition function. That is, $f(\epsilon) \in Q_0$, and for every finite word $u \in \Sigma^*$ and letter $\sigma \in \Sigma$, we have that $f(u \cdot \sigma) \in \delta(f(u),\sigma)$. 
	\item
	Following $f$ causes $\A$ to accept all the words in its language. That is, for every word $w = \sigma_1 \cdot \sigma_2 \cdots$, if $w \in L(\A)$, then the run $f(\epsilon), f(w[1, 1])$, $f(w[1, 2]), \ldots$ %, which we denote by $f(w)$, 
	is an accepting run of $\A$ on $w$. 
\end{enumerate}
We say that the strategy $f$ \emph{witnesses} $\A$'s HDness. 
Note that, by definition, we can assume that every SD and HD automaton $\A$ has a single initial state. Thus, we sometimes abuse notation and write $\A$ as $\A = \zug{\Sigma, Q, q_0, \delta, \alpha}$, where $q_0$ is the single initial state of the SD (or HD) automaton $\A$.

For an automaton $\A$, we say that a state $q$ of $\A$ is \emph{HD}, if $\A^q$ is HD. Note that every deterministic automaton is HD. Also, while not all HD automata can be pruned to deterministic ones \cite{BKKS13}, removing of transitions that are not used by a strategy $f$ that witnesses $\A$'s HDness does not reduce the language of $\A$ and results in an SD automaton. Moreover, since every state that is used by $f$ is HD, the removal of non-HD states does not affect $\A$'s language nor its HDness. Accordingly, we have the following~\cite{KS15, BK18}. 

\begin{proposition}\label{gfg is SD}
	Every HD automaton $\A$ can be pruned to an equivalent SD automaton all whose states are HD.
\end{proposition}

We use three-letter acronyms in $\{\text{D}, \text{N}\} \times \{\text{B}, \text{C}, \text{W}, \text{F}\}\times \{\text{W}\}$ to denote the different automata classes. The first letter stands for the branching mode of the automaton (deterministic or nondeterministic); the second for the acceptance condition type (\buchi, co-\buchi, weak, or an automaton that runs on  finite inputs); and the third indicates that we consider automata on words. For transition-based automata, we start the acronyms with the letter ``t",  and for HD or SD automata, we add an HD or SD prefix, respectively. For example, an HD-tNBW is a transition-based HD nondeterministic \buchi automaton, and an SD-NWW is a state-based weak SD automaton.

\stam{
For a tNCW $\A = \zug{\Sigma, Q, q_0, \delta, \alpha}$, we say that a run $r$ of $\A$ is \emph{safe} if it does not traverse $\alpha$-transitions, and we use the term \emph{safe-transition} to refer to a transition not in $\alpha$.
}%end of stam

%end of stam

\section{Semantically Deterministic \buchi Automata}

In this section we examine SD-tNBWs and SD-NBWs. Our results use the following definitions and constructions:
For a language $R \subseteq \Sigma^*$ of finite words, we use $\im{R}$ to denote the language of infinite words that contain infinitely many disjoint infixes in $R$. Thus, $w \in  \im{R}$ iff $\epsilon\in R$ or there are infinitely many indices $i_1 \leq i'_1 < i_2 \leq i'_2 < \cdots$ such that $w[i_j,i'_j] \in R$, for all $j \geq 1$. For example, taking $\Sigma=\{a,b\}$, we have that $\im{\{ab \}}$ is the language of words with infinitely many $ab$ infixes, namely all words with infinitely many $a$'s and infinitely many $b$'s.  We say that a finite word $x\in \Sigma^*$ is a {\em good prefix} for a language $R \subseteq \Sigma^*$ if for all finite words $y\in \Sigma^*$, we have that $x\cdot y\in R$. For example, while the language $(a+b)^* \cdot a$ does not have a good prefix, the word $a$ is a good prefix for the language $a \cdot (a+b)^*$. 

Theorem \ref{NFW to SD-tNBW operation thm} below states that one can translate an NFW that recognizes a language $R \subseteq \Sigma^*$ to an SD-tNBW for $\im{R}$. Theorem \ref{SD-tNBW to NFW operation thm} is more complicated and suggests that this translation can be viewed as a reversible encoding, in the sense that one can obtain an NFW for $R$ from an SD-tNBW for $\im{(\$\cdot R \cdot \$)}$, where $\$\notin \Sigma$. The translations involve no blow up, and both theorems play a major rule in the rest of this section.

\begin{theorem}\label{NFW to SD-tNBW operation thm}
	Given an NFW $\N$, one can obtain, in polynomial time, an SD-tNBW $\A$ such that $L(\A)=\infty L(\N)$ and $|\A|=|\N|$.
\end{theorem}

\begin{proof}
	Let $\N  = \zug{\Sigma, Q, Q_0, \delta, F}$. Then, $\A = \zug{\Sigma, Q, Q_0, \delta', \alpha}$ is obtained  from $\N$ by adding transitions to $Q_0$ from all states with all letters. A new $\sigma$-transition is in $\alpha$ if $Q_0 \cap F \neq \emptyset$ or when $\N$ could transit with $\sigma$ to a state in $F$. Formally, for all $s \in Q$ and $\sigma \in \Sigma$, we have that $\delta'(s,\sigma)=\delta(s,\sigma) \cup Q_0$, and $\alpha =  \{\zug{s, \sigma, q}: q\in Q_0 \text{ and } (Q_0\cap F \neq \emptyset  \mbox{ or } \delta(s, \sigma) \cap F \neq \emptyset)\}$.
	
	It is easy to see that $|\A|=|\N|$. In order to prove that $\A$ is SD and $L(\A)=\infty L(\N)$, we prove next that for every state $q\in Q$, it holds that $L(\A^q)=\infty L(\N)$. 
	We distinguish between two cases. First, if $Q_0\cap F \neq \emptyset$, then $\epsilon \in L(\N)$, in which case $\infty L(\N) = \Sigma^\omega$. Indeed, every word in  $\Sigma^\omega$ has infinitely many empty infixes. Recall that when $Q_0\cap F\neq \emptyset$,  the tNBW $\A$ has $\Sigma$-labeled $\alpha$ transitions from every state to $Q_0$. Then, the run $q, (q_0)^\omega$, where $q_0\in Q_0$, is an accepting run of $\A^q$ on any word $w\in \Sigma^\omega$. So, $L(\A^q)  = \Sigma^\omega = \infty L(\N)$, and we are done. 
	
	Now, if $Q_0\cap F = \emptyset$, we first prove that for every state $q\in Q$, we have that $\infty L(\N)\subseteq L(\A^q)$. Consider a word $w \in \infty L(\N)$, and consider a run of $\A^q$ on $w$ that upon reading an infix $\sigma_1 \cdot u \cdot \sigma_2$ such that $u \cdot \sigma_2 \in  L(\N^{q_0})\subseteq L(\N)$, for some $q_0\in Q_0$, moves to $q_0$ while reading the letter $\sigma_1$, and then follows an accepting run of $\N$ on $u \cdot \sigma_2$ except that, upon reading the letter $\sigma_2$, it moves to a state in  $Q_0$, traversing an $\alpha$ transition, instead of traversing a transition to $F$. Note that the transitions in $\A$ enable such a behavior. Also note that $\infty L(\N) \subseteq \infty \Sigma \cdot L(\N)$, and so such a run traverses $\alpha$ infinitely often and is thus accepting.
	
	It is left to prove that $L(\A^q)\subseteq \infty L(\N)$ for every state $q \in Q$.
	Let $r=r_0, r_1, \ldots$ be an accepting run of $\A^q$ on a word $w =  \sigma_1\cdot  \sigma_2\cdots$. We show that $w \in \infty L(\N)$. For that, we show that for every $i \geq 0$, there are $k_1,k_2$ such that $i \leq k_1 \leq k_2$ and $w[k_1+1, k_2+1]$ is an infix of $w$ in $L(\N)$. 
	Being accepting, the run $r$ traverses infinitely many $\alpha$-transitions. Given $i \geq 0$, let $k_2 \geq i$ be such that 
	$t=\zug{r_{k_2}, \sigma_{k_2+1}, r_{k_2+1}}$ is an $\alpha$-transition, but not the  first $\alpha$-transition that $r$ traverses when it reads the suffix $w[i+1,\infty]$, and let $k_1$ be the maximal index such that $i \leq k_1\leq  k_2$ and $r_{k_1} \in Q_0$. Since $t$ is not the first $\alpha$-transition that $r$ traverses when it reads the suffix $w[i+1,\infty]$, and  since $\alpha$-transitions lead to $Q_0$, then $k_1$ exists. Note that the maximality of $k_1$ imply that the run  $r_{k_1}, r_{k_1 + 1}, \ldots, r_{k_2}$ is a run of $\N$ on the infix $w[k_1+1, k_2]$. Now, as $\zug{r_{k_2}, \sigma_{k_2+1}, r_{k_2+1}} \in \alpha$ and $Q_0\cap F = \emptyset$, it follows by the definition of $\delta'$  that there is a state $s \in \delta(r_{k_2}, \sigma_{k_2+1}) \cap F$. Hence, $r_{k_1}, r_{k_1 + 1}, \ldots, r_{k_2}, s$ is an accepting run of $\N$ on the infix $w[k_1+1, k_2+1]$, which is thus in $L(\N)$. 
\end{proof}

We proceed to the other direction, translating an automaton for $\infty (\$\cdot R\cdot \$ )$ to an automaton for $R$. Note that the language $R$ is well defined. That is, if $R_1 \subseteq \Sigma^*$ and $R_2 \subseteq \Sigma^*$ are such that $ \infty (\$\cdot R_1\cdot \$ )= \infty (\$\cdot R_2\cdot \$ )$, then it must be that $R_1=R_2$. 

\begin{theorem}\label{SD-tNBW to NFW operation thm}
	Consider a language $R \subseteq \Sigma^*$ and a letter $\$\notin \Sigma$. For every SD-tNBW $\A$ such that \mbox{$L(\A)=\infty (\$\cdot R\cdot \$ )$}, there exists an NFW $\N$ such that $L(\N) = R$ and $|\N|\leq |\A| + 1$. In addition, if $R$ has no good prefixes, then $|\N| \leq |\A|$.
\end{theorem}

\begin{proof}

If $R$ is trivial, then one can choose $\N$ to be a one-state NFW. Assume that $R$ is nontrivial. Let $\A = \zug{\Sigma\cup \{\$\}, Q, q_0, \delta, \alpha}$ be an SD-tNBW for $\infty (\$ \cdot R \cdot \$)$,   W.l.o.g we assume that all the states of $\A$ are reachable. For a nonempty set of states $S\in 2^Q \setminus \emptyset$, we
define the \emph{universal $\overline{\alpha}$ language} of $S$ as $L_{\it u\overline{\alpha}}(S) = \{ w\in (\Sigma \cup \{\$\})^\omega: \text{for all  $q\in S$, all the runs of $\A^q$ on $w$ do not traverse $\alpha$}\}$.
We say that $S$ is \emph{hopeful} when $(\$\cdot \overline{R})^\omega \subseteq L_{\it u\overline{\alpha}} $.
Note that $S$ is hopeful iff for every state $q\in S$, it holds that $\{q\}$ is hopeful. %\color{red} Hopeful sets $S$ are hopeful in the sense that they try to capture words in $R$: if $S$ is hopeful, $x\in \Sigma^*$ is a finite word, and there is a run from $S$ on $\$\cdot x$ that traverses $\alpha$, then $x\in R$, yet hopeful sets do not necessarily charecterize the language $R$ as it could be the case that there is a word $x\in R$ such that all the runs of $S$ on $\$\cdot x$ do not traverse $\alpha$.  Hence, we introduce the notion of \emph{good} sets, and \color{black}
Also, if $S$ is hopeful, $x\in \Sigma^*$ is a finite word, and there is a run from $S$ on $\$\cdot x$ that traverses $\alpha$, then $x\in R$.
Then, we say that $S$ is \emph{good} when for all words $x\in \Sigma^*$, it holds that $x\in \overline{R}$ iff all the runs from $S$ on $\$\cdot x$ do not traverse $\alpha$, and the set $\delta(S, \$\cdot x)$ is hopeful.
Note that as $R$ is nontrivial, there exists a word $x$ in $\overline{R}$, and thus by definition, all the $\$$-labled transitions going out from a good set $S$ are in $\overline{ \alpha}$.

Note that good sets in $\A$ characterize the language $R$. We first show that a good set exists. 
Assume towards contradiction that every nonempty set $S\in 2^Q\setminus \emptyset$ is not good. We define iteratively a sequence $S_1, S_2, S_3, \ldots $ of nonempty sets in $2^Q$, and a sequence $x_1, x_2, x_3,  \ldots$ of finite words in $(\Sigma\cup \{\$\})^+$, as follows. The set $S_1$ is chosen arbitrarily. For all $i\geq 1$, given the set $S_i$, we define $x_i$ and $S_{i+1}$ as follows.  
By the assumption, $S_i$ is not good. Therefore at least one of the following holds:
\begin{enumerate}[label=(\alph*)]
	\item  
	There is a witness in $R$ that $S_i$ is not good: there is a word $a_i\in R$ such that all the runs from $S_i$ on $\$\cdot a_i$ do not traverse $\alpha$, and the set $\delta(S_i, \$\cdot a_i)$ is hopeful.
	
	\item 
	There is a witness in $\overline{R}$ that $S_i$ is not good: there is a word $b_i\in \overline{ R}$ such that there is a run from $S_i$ on $\$\cdot b_i$ that traverses $\alpha$, or the set $\delta(S_i, \$\cdot b_i)$ is not hopeful.
	
\end{enumerate}
If there is a witness $a_i\in R$ that $S_i$ is not good, we take $x_i = \$\cdot  a_i$ and $S_{i+1} = \delta(S_i, x_i)$. 
Otherwise, there is a witness $b_i\in \overline{ R}$ that $S_i$ is not good. Then, as we argure below, there exists a word $z_i \in (\$\cdot \overline{ R})^+$, 
and there is a run $r_i$ from $S_i$ on $z_i$ that traverses $\alpha$. In this case, we take $S_{i+1}$ to be the set that contains only the state that $r_i$ ends in, and we take $x_i = z_i$. 
To see why such $z_i$ exists, we distinguish between two cases. First, 
if there is a run from $S_i$ on $\$\cdot b_i$ that traverses $\alpha$, then we take $z_i= \$\cdot b_i$. Otherwise, the set $\delta(S_i, \$\cdot b_i)$ is not hopeful. Therefore, it contains a state $q$ that has a run that traverses $\alpha$ on a word in $ (\$\cdot \overline{ R})^\omega$. The latter implies that there is a word $z' \in (\$\cdot \overline{ R})^+$ such that there is a finite run $r'$ of $q$ on $z'$ that traverses $\alpha$. Then, since $q$ is reachable from $S_i$ upon reading $\$\cdot b_i$, we take $z_i = \$\cdot b_i \cdot z'$. Indeed, a run $r_i$ from $S_i$ on $z_i= \$\cdot b_i \cdot z'$ that traverses $\alpha$, starts from $S_i$ and reaches $q$ upon reading $\$\cdot b_i$, and then follows the run $r'$ upon reading $z'$.

So, let $S_1, S_2, S_3 \ldots$ and $x_1, x_2, x_3, \ldots $ be sequences as above.
We distinguish between three cases:
\begin{enumerate}
	\item 
	There exists $i \geq 1$ such that for all $j\geq i$, there is a witness $a_j\in R$ that $S_j$ is not good: 
	consider $j\geq i$. In this case, $x_j = \$ \cdot a_j \in \$\cdot R$, $S_{j+1} = \delta(S_j, x_j)$, and all the runs of $S_j$ on $x_j$ do not traverse $\alpha$.
	Hence, all the runs from $S_i$ on the word $w = x_i\cdot x_{i+1} \cdot x_{i+2} \cdots \in (\$\cdot R)^\omega$ do not traverse $\alpha$, in particular, if we consider a state $q\in S_i$, we get that the word $w$ is not accepted from $q$.
	As all the states of $\A$ are reachable, there is a word $y\in (\Sigma\cup \{\$\})^*$ such that $q\in \delta(q_0, y)$. Then, the SDness property of $\A$ implies that all the states in $\delta(q_0, y)$ do not accept $w$, and hence 
	the word $y\cdot w$ is not accepted by $\A$ even though it has infinitely many infixes in $\$\cdot R \cdot \$$, and we have reached a contradiction. 
	
	\item 
	There exists $i \geq 1$ such that for all $j\geq i$, no word in  $R$ witnesses that $S_j$ is not good: consider $j\geq i$. In this case, as argued above, there is a witness $b_j\in \overline{ R}$ that $S_j$ is not good. So,  $x_j = z_j \in (\$\cdot \overline{ R})^+$, and there is a run $r_j$ from $S_j$ on $z_j$ that traverses $\alpha$ and ends in a state $s_{j+1}$, where $S_{j+1} = \{s_{j+1}\}$. The concatenation of the runs $r_i, r_{i+1}, \ldots$ is an accepting run of $\A^q$, for some $q\in S_i$, on the word $z_i\cdot z_{i+1}  \cdot z_{i+2}\cdots \in (\$\cdot \overline{ R})^\omega$. Let $y\in (\Sigma\cup \{\$\})^*$ be such that $q\in \delta(q_0, y)$; then the word $w = y\cdot z_i\cdot z_{i+1} \cdot z_{i+2} \cdots$ is accepted by $\A$ even though it has only finitely many infixes in $\$\cdot R\cdot \$$, and we have reached a contradiction.

	\item If the above two cases do not hold, then there are infinitely many indices $i\geq 1$ such that no word in $R$ witnesses that $S_i$ is not good, and there are infinitely many indices $i\geq 1$ such that there is a witness in $R$ that $S_i$ is not good. Therefore, there must be $i\geq 1$ such that there is a witness $a_i\in R$ that $S_i$ is not good, and no word in $R$ witnesses that $S_{i+1}$ is not good.
	On the one hand, the set $S_{i+1} = \delta(S_i, \$\cdot a_i)$ is hopeful. On the other hand, the fact that no word in $R$ witnesses that $S_{i+1}$ is not good, implies that there is a witness $b_{i+1}\in \overline{ R}$ that $S_{i+1}$ is not good, and thus there 
	is a run from $S_{i+1}$ that traverses $\alpha$ on the word $x_{i+1} = z_{i+1}\in (\$\cdot \overline{ R})^+$. Thus, we have reached a contradiction to the fact that $S_{i+1}$ is hopeful.
\end{enumerate}

We show now that a good set in $\A$ induces an NFW $\N$ for $R$ with the required properties. Let $S\in 2^Q\setminus\emptyset$ be a good set. We define the NFW $\N = \zug{\Sigma, Q\cup \{q_{acc}\}, Q^S_0, \delta_S, F_S}$, where $Q^S_0 = \delta(S, \$)$, $F_S = \{q_{acc}\} \cup \{q\in Q: \text{the set $\{q\}$ is not hopeful} \}$, and the transition function $\delta_S$ is defined as follows. For every two states $q, s\in Q$ and letter $\sigma\in \Sigma$, it holds that $s\in \delta_S(q, \sigma)$ iff $\zug{q, \sigma, s}\in \overline{ \alpha}$. Also, $q_{acc}\in \delta_S(q, \sigma)$ 
iff there is a $\sigma$-labeled $\alpha$-transition going out from $q$ in $\A$. Also, for all letters $\sigma \in \Sigma$, it holds that $\delta(q_{acc}, \sigma) = \{q_{acc}\}$; that is, $q_{acc}$ is an accepting sink.
Thus, $\N$ behaves as the states in $\delta(S, \$)$ as long as it reads $\overline{\alpha}$ transitions of $\A$, moves to the accepting sink $q_{acc}$ whenever an $\alpha$-transition is encountered, and accepts also whenever it reaches a state in $Q$ that is not hopeful.

%In Appendix~\ref{app SD-tNBW to NFW operation thm}, we prove that $L(\N) = R$. 
%Essentially, this follows from the fact that if we consider a word $x\in \Sigma^*$ such that all the runs from $Q^S_0$ on $x$ in $\A$ do not traverse $\alpha$, then $S$ being a good set implies that $x\in R$ iff $\delta(S, \$\cdot x) \cap F_S \neq \emptyset$.
%
We prove next that $L(\N) = R$. We first prove that $L(\N)\subseteq R$. Consider a word $x= \sigma_1\cdot \sigma_2\cdots \sigma_n \in L(\N)$. 
If there is a state $q\in \delta(S, \$)$ such that $\A^q$ has a run on $x$ that traverses $\alpha$, then as $S$ is good, we get that $x\in R$. 
Otherwise, for all states $q\in \delta(S, \$)$, it holds that all the runs of $\A^q$ on $x$ do not traverse $\alpha$.
Let $r = r_0, r_1, \ldots r_n$ be an accepting run of $\N$ on $x$.
By the definition of $\N$, we have that $r_0\in \delta(S, \$)$; in particular, all the runs of $\A^{r_0}$ on $x$ do not traverse $\alpha$. Hence,
since $r$ follows a run of $\A^{r_0}$ as long it does not visit $q_{acc}$, and since $r$ visits $q_{acc}$ only when it follows a run of $\A^{r_0}$ that traverses $\alpha$, we get that the run $r$ exists in $\A$. Hence, $r_n\in Q$ and as $r$ is an accepting run of $\N$, we get that the set $\{r_n\}$ is not hopeful. The fact that $r$ exists in $\A$ implies also that  $r_n\in \delta(S, \$\cdot x)$, and so $\delta(S, \$\cdot x)$ is not hopeful. Hence, as $S$ is a good set, we conclude that $x\in R$.

For the other direction, namely $R\subseteq L(\N)$, consider a word $x\in R$. If there is a state $q\in \delta(S, \$)$ and a run $r$ of $\A^q$ on $x$ that traverses $\alpha$, then an accepting run $r'$ of $\N$ on $x$ can be obtained by following the run $r$, and moving to the accepting sink $q_{acc}$ when $r$ traverses $\alpha$ for the first time. Indeed, $q\in Q^S_0$,  $\overline{\alpha}$ transitions of $\A$ exist in $\N$, and $\N$ moves to $q_{acc}$ whenever $\A$ traverses $\alpha$.
Otherwise, for all $q\in \delta(S, \$)$, it holds that all the runs of $\A^q$ on $x$ do not traverse $\alpha$. Also, recall that all the $\$$-labeled transitions from a good set are in $\overline{\alpha}$. Hence, all the runs from $S$ on $\$\cdot x$ do not traverse $\alpha$. In this case, as $S$ is a good set, and $x\in R$, it follows that $\delta(S, \$\cdot x)$ is not hopeful. Therefore, there is a state $s\in \delta(S, \$\cdot x)$ such that $\{s\}$ is not hopeful. Hence, by the definition of $\N$, it holds that $s\in F_S$. Hence, a run in $\A$ from $\delta(S, \$)$ that reaches $s$ upon reading $x$ and does not traverse $\alpha$, is an accepting run of $\N$ on $x$, and we are done.

Since the state space of $\N$ is $Q\cup \{q_{acc}\}$, then $|\N| = |\A|+1$. Moreover, 
as $q_{acc}$ is an accepting sink, a word $x\in L(\N)$ that has a run that ends in $q_{acc}$ is  a good prefix for $L(\N)$. Hence, as $L(\N) = R$, if $R$ has no good prefixes, then $q_{acc}$ is not reachable in $\N$ and thus can be removed without affecting $\N$'s language. Thus, in this case, we get an NFW $\N$ for $R$ whose size is at most $|\A|$. 
\end{proof}

\subsection{Succinctness and Complementation}\label{SD-tNBW succinct sec}

In this section we study the succinctness of SD B\"uchi automata with respect to deterministic ones, and the blow-up involved in their complementation. 
We show that SD-tNBWs are exponentially more succinct than tDBWs, matching the known upper bound \cite{AKL21}, and in fact, also from HD-tNBWs. 
We also prove an exponential lower bound for complementation.
Similar results for SD-NBWs follow, as the transition between the two types of acceptance conditions is linear.

\begin{theorem}\label{SD-tNBW succinct thm}
	There is a family $L_1,L_2,L_3,\ldots$ of languages such that for every $n \geq 1$, there is an SD-tNBW with $3n+3$ states that recognizes $L_n$, yet every tDBW or HD-tNBW that recognizes $L_n$ needs at least $2^n$ states.
\end{theorem}

\begin{proof}
	For $n \geq 1$, let $[n]=\{1,\ldots,n\}$, and let $\Sigma_n = \{1,\ldots,n,\$, \#\}$. 
	We say that a word $z \in \Sigma^*_n$ is \emph{good} if $z = \$\cdot  x  \cdot \#  \cdot i$, where $x\in [n]^+$ and $i$ appears in $x$. Let $R_n \subseteq \Sigma^*_n$ be the language of all good words. 
	We define $L_n=\infty R_n$.
	First, it is not hard to see that $R_n$ can be recognized by an NFW $\N_n$ with $3n+3$ states. 
	Essentially, $\N_n$ guesses the last letter $i$ in the input word and then checks that the guess is correct.
	A sketch of $\N_n$ appears in Figure~\ref{N_n} below. 
	By Theorem~\ref{NFW to SD-tNBW operation thm}, there is an SD-tNBW for $L_n$ with $3n+3$ states.

	\begin{figure}[htb]	
		\begin{center}
			\includegraphics[width=0.5\textwidth]{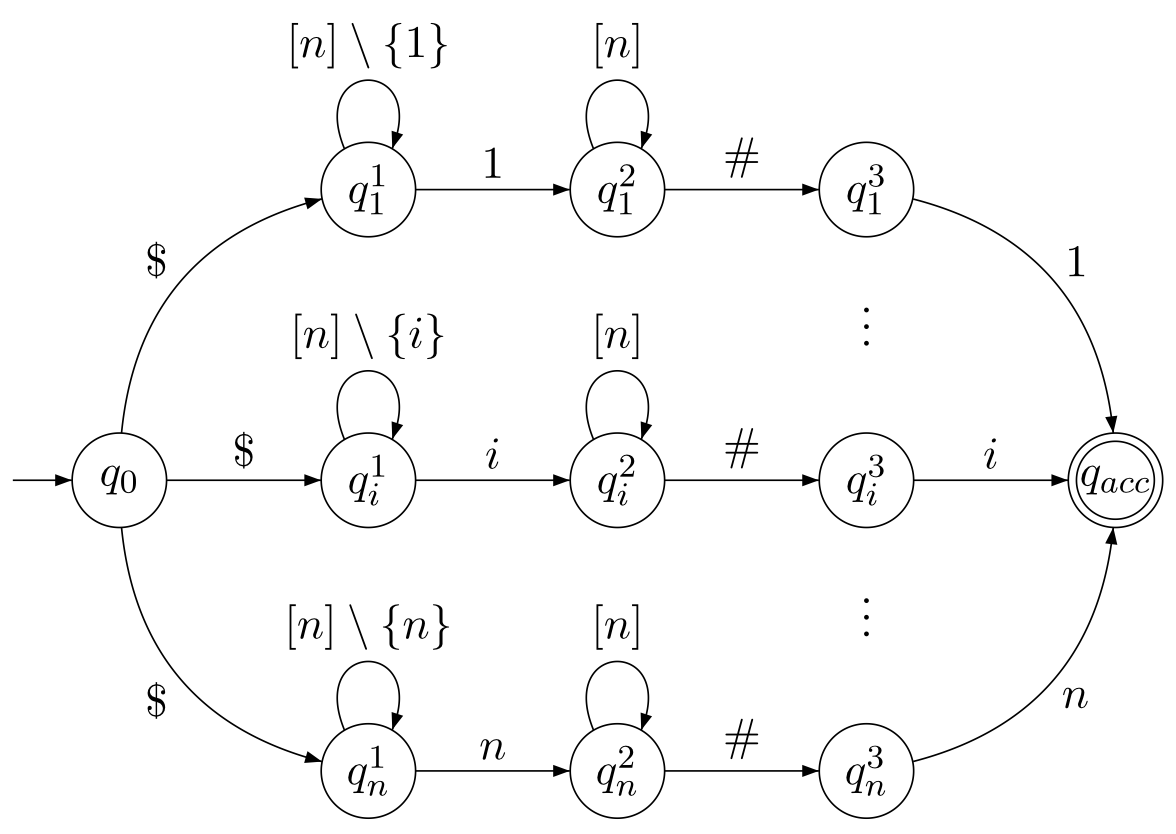}
			\caption{The NFW $\N_n$. Missing transitions lead to a rejecting-sink}
			\label{N_n}
		\end{center}
	\end{figure}

	Before we prove that every tDBW or HD-tNBW that recognizes $L_n$ needs at least $2^n$ states, let us note that it is already known that going from a DFW for a language $R \subseteq \Sigma^*$ to a tDBW
	for $\infty R$, may involve a blow-up of $2^{n-2 - log_2 (n)}$\cite{KL20}. Thus, Theorem~\ref{NFW to SD-tNBW operation thm} implies 
	an exponential gap between SD-tNBWs and tDBWs.
	Moreover, as HD-tNBWs are at most quadratically more succinct than tDBWs~\cite{KS15}, the above can be extended to a $2^{\frac{n- 2 - log_2(n)}{2}}$ lower bound for the succinctness of SD-tNBWs with respect to HD-tNBWs. 
	Our example here is tighter. 
	
	We now prove that every HD-tNBW that recognizes $L_n$ needs at least $2^n$ states. Assume towards contradiction that $\A = \zug{\Sigma_n,Q, q_0, \delta, \alpha}$  is an HD-tNBW for $L_n$ with $|Q|<2^n$ states. 
	Below, we
	iteratively define infinite sequences of finite words $x_1, x_2, x_3,\ldots $ and states $q_0, q_1, \ldots$ such that for all $k \geq 1$,
	the word $x_k$ starts with $\$$, has no good infixes, and there is a run of the form $r_k = q_{k-1} \xrightarrow{x_k} q_k$ in $\A$ on $x_k$ that traverses $\alpha$. To see why such sequences imply a contradiction, note that the concatenation of the runs $r_1, r_2, \ldots$ is an accepting run of $\A$  on the word $x = x_1\cdot x_2\cdots$. As $\A$ recognizes $L_n = \infty R_n$, it follows that $x\in \infty R_n$. On the other hand, $x$ has no good infixes,  and so $x\notin \infty R_n$. Indeed, if there is a good infix in $x$, then it must contain letters from different $x_k$'s; in particular, it must contain at least two $\$$'s. 

	We proceed by describing how, given $x_1, x_2, \ldots, x_{k-1}$ and $q_0, q_1, \ldots, q_{k-1}$, we define $x_k$ and $q_k$.  
	The challenging part in the construction is to make it valid also for HD (and not only deterministic) automata. For this, the definition of the words in the sequence is defined with respect to a strategy that attempts to witness the HDness of $\A$.
	First, recall that by Proposition~\ref{gfg is SD}, we can assume that $\A$ is SD and all its state are HD. Also, note that $q_0 \xrightarrow{x_1}  q_1  \xrightarrow{x_2} q_2 \cdots \xrightarrow{x_{k-1}} q_{k-1}$ is a run in $\A$; in particular, $q_{k-1}$ is reachable in $\A$. Hence, $L(\A) \subseteq L(\A^{q_{k-1}})$. Indeed, if $y$ is a word in $L(\A) = \infty R_n$, then $x_1\cdot x_2 \cdots x_{k-1} \cdot y\in \infty R_n$, and so the SDness of $\A$ implies that all the states in $\delta(q_0, x_1\cdot x_2 \cdots x_{k-1})$; in particular $q_{k-1}$, accept $y$.
	We abuse notation and use $S$, for a subset $S\subseteq [n]$, to denote a finite word over $[n]$ consisting of exactly the numbers appearing in $S$.
	We can now define $x_k$ and $q_k$.
	Fix $u_0, d_0 = q_{k-1}$, and let $f$ be a strategy witnessing $\A^{q_{k-1}}$'s HDness. We define inductively two infinite runs $r_u$ and $r_d$ on distinct words, and of the form $$r_u = u_0 \xrightarrow{\$\cdot 2^{[n]}} p_0 \xrightarrow{\#\cdot [n]} u_1 \xrightarrow{\$\cdot 2^{[n]}} p_1 \xrightarrow{\#\cdot [n]} u_2 \xrightarrow{\$\cdot 2^{[n]}} \cdots $$ $$r_d = d_0 \xrightarrow{\$\cdot 2^{[n]}} p_0 \xrightarrow{\#\cdot [n]} d_1 \xrightarrow{\$\cdot 2^{[n]}} p_1 \xrightarrow{\#\cdot [n]} d_2 \xrightarrow{\$\cdot 2^{[n]}} \cdots $$ Thus, while the states $u_i$ and $d_i$ are not necessarily equal, the runs meet infinitely often at the states $p_i$ upon reading infixes of the form $\$\cdot 2^{[n]}$. When we define the runs, we make sure that both are runs of $\A^{q_{k-1}}$ that are consistent with the strategy $f$.
	Consider $j\geq 0$, and assume that we have defined the following sub-runs $u_0 \xrightarrow{z^j_u} u_j$ and $d_0 \xrightarrow{z^j_d} d_j$ that are consistent with the strategy $f$. 
	A finite run $r$ of the state $u_j$ on a word $z = z_1\cdot z_2 \cdots z_n$ is \emph{consistent} if it is of the form the form $r = f(z^j_u), f(z^j_u \cdot z_1), \ldots, f(z^j_u z_n)$; that is, $r$ is consistent with the strategy $f$ w.r.t the prefix $z^j_u$. Similarly, we consider runs of $d_j$ that are consistent with the strategy $f$ w.r.t the prefix $z^j_d$.  
	As $|Q| < 2^n$, there are subsets $S_{j+1}, T_{j+1} \subseteq [n]$ with $i_{j+1}\in S_{j+1}\setminus T_{j+1}$ such that 
	$$((u_j \xrightarrow{\$\cdot S_{j+1}} p_j \xrightarrow{\#\cdot i_{j+1}} u_{j+1}) \mbox{ and } (d_j \xrightarrow{\$\cdot T_{j+1}} p_j \xrightarrow{\#\cdot i_{j+1}} d_{j+1}  ))$$ $$\mbox{ or }$$ $$((u_j \xrightarrow{\$\cdot T_{j+1}} p_j \xrightarrow{\#\cdot i_{j+1}} u_{j+1})  \mbox{ and } (d_j \xrightarrow{\$\cdot S_{j+1}} p_j \xrightarrow{\#\cdot i_{j+1}} d_{j+1} ) )$$
	where the runs are consistent. 
	Indeed, if we consider all consistent runs of $u_j$ on words of the form $\$\cdot 2^{[n]}$ and all consistent runs of $d_j$ on words of the form $\$\cdot  2^{[n]}$, we get that there are two runs, of $u_j$ and of $d_j$, corresponding to distinct subsets of $[n]$,  that reach the same state $p_j$.  
	So far, we have defined two infinite runs of $\A^{q_{k-1}}$ that are consistent with the strategy $f$:
	
	$$r_u = u_0 \xrightarrow{\$\cdot U_1} p_0 \xrightarrow{\#\cdot i_1} u_1 \xrightarrow{\$\cdot U_2} p_1 \xrightarrow{\#\cdot i_2} u_2 \xrightarrow{\$\cdot U_3} \cdots $$ $$r_d = d_0 \xrightarrow{\$\cdot D_1} p_0 \xrightarrow{\#\cdot i_1} d_1 \xrightarrow{\$\cdot D_2} p_1 \xrightarrow{\#\cdot i_2} d_2 \xrightarrow{\$\cdot D_3} \cdots $$
	
	where for all $j\geq 1$, $i_j \in (U_j \setminus D_j) \cup (D_j \setminus U_j)$. In particular, exactly one of $\$\cdot U_j \cdot \#\cdot i_j$ and $\$\cdot D_j \cdot \#\cdot i_j$ is good. Therefore, at least one of the runs $r_u$ and $r_d$ runs on a word in $\infty R_n$. Since both runs are consistent with the strategy $f$ and $\infty R_n \subseteq L(\A^{q_{k-1}})$, we get that one of the runs traverses $\alpha$ infinitely often.
	W.l.o.g consider some $j > 0$ such that the sub-run $p_j \xrightarrow{\#_{i_{j+1}}} u_{j+1} \xrightarrow{\$\cdot U_{j+2}} p_{j+1}$ traverses $\alpha$. 
	We define 
	$$x_k = (\$ \cdot A_1 \cdot \# \cdot i_1) \cdot (\$ \cdot A_2 \cdot \# \cdot i_2) \cdots (\$ \cdot A_j \cdot \# \cdot i_j) \cdot (\$\cdot A_{j+1} \cdot \# \cdot {i_{j+1}}) \cdot \$\cdot U_{j+2}$$
	where for all $1\leq l \leq j+1$, it holds that $A_l \in \{U_l, D_l\}$ and $i_l \notin A_l$. In particular, $x_k$ starts with $\$$ and has no good infixes. To conclude the proof, we show that there is a run $r_k$ of $\A^{q_{k-1}}$ on $x_k$ that traverses an $\alpha$-transition (in particular, we can choose $q_k$ to be the last state that $r_k$ visits). The run $r_k$ can obtained from the runs $r_u$ and $r_d$ as follows. First, $r_k$ starts by following one of the sub-runs $u_0 \xrightarrow{\$\cdot U_1} p_0$ or $d_0 \xrightarrow{\$\cdot D_1} p_0$, depending on whether $A_1 = U_1$ or $A_1 = D_1$, respectively. Then, if the run $r_k$ is in a state of the form $p_m$, for $m < j$, the run $r_k$ proceeds, upon reading $\#\cdot i_{m+1}$, to one of the states $u_{m+1}$ or $d_{m+1}$, depending on whether the next two letters to be read are $\$\cdot U_{m+2}$ or $\$\cdot D_{m+2}$, respectively. Finally, once $r_k$ reaches $p_j$, it ends by traversing $\alpha$ while following the sub-run $p_j \xrightarrow{\#_{i_{j+1}}} u_{j+1} \xrightarrow{\$\cdot U_{j+2}} p_{j+1}$.
\end{proof}

\begin{remark}\label{alt min proof remark}
	{\rm In order to get a slightly tighter bound, one can show that a minimal tDBW for $L_n$ needs at least $2^{n+1}$ states and that the language $L_n$ is not {\em HD-helpful}. That is, a minimal HD-tNBW for $L_n$ is not smaller than a minimal tDBW for $L_n$, and so the $2^{n+1}$ bound holds also for a minimal HD-tNBW. The proof starts with a tDBW for $L_n$ that has $2^{n+1}$ states, considers its complement tDCW, and shows that the application of the polynomial minimization algorithm of \cite{AK19} on it, namely the algorithm that returns an equivalent minimal HD-tNCW, does result in a smaller automaton. The result then follows from the fact that a minimal HD-tNCW for a language is smaller than a minimal HD-tNBW  for its complement \cite{KS15}. The proof is easy for readers familiar with \cite{AK19}, yet involves many notions and observations from there. We thus leave it for the appendix. Specifically,   
	in Appendix~\ref{app alt min proof remark}, we describe a tDBW with $2^{n+1}$ states for $L_n$, and readers familiar with \cite{AK19} can observe that the application of the HD-tNCW minimization algorithm on its dual tDCW does not make it smaller. 
	}\hfill \qed
\end{remark}

\begin{theorem}\label{SD-tNBW comp thm}
	There is a family $L_1,L_2,L_3,\ldots$ of languages such that for every $n \geq 1$, there is an SD-tNBW with $O(n)$ states that recognizes $L_n$, yet every SD-tNCW that recognizes $\overline{L_n}$ needs at least $2^{O(n)}$ states.
\end{theorem}

\begin{proof}
Let $\Sigma=\{0,1\}$. For $n \geq 1$, let $R_n=\{w:  w \in 0 \cdot (0+1)^{n-1} \cdot 1 + 1 \cdot (0+1)^{n-1} \cdot 0\}$. Thus, $R_n$ contains all words of length $n+1$ whose first and last letters are different. 
	It is easy to see that $R_n$ can be recognized by an NFW (in fact, a DFW) $\N_n$ with $O(n)$ states. We define $L_n=\infty R_n$. First, by Theorem~\ref{NFW to SD-tNBW operation thm}, there is an SD-tNBW with $O(n)$ states for $L_n$.
	In order to prove that an SD-tNCW for $\overline{L_n}$ needs at least $2^{O(n)}$ states, we prove that in fact every tNCW for $\overline{L_n}$ needs that many states. For this, note that $\overline{L_n}$ consists of all words $w$ for which there is $u \in (0+1)^n$ such that $w \in (0+1)^* \cdot u^\omega$. Indeed, for such words $w$, the suffix $u^\omega$ contains no infix in $R_n$. Also, if a word contains only finitely many infixes in $R_n$, then it must have a suffix with no infixes in $R_n$, namely a suffix of the form $u^\omega$ for some $u \in (0+1)^n$. Then, the proof that a tNCW for $\overline{L_n}$ needs exponentially many states is similar to the proof that an NFW for $\{u \cdot u : u \in (0+1)^n\}$ needs exponentially many states. Indeed, it has to remember the last $n$ letters read. 
\end{proof}

%end of stam

\subsection{Decision Problems}\label{dp buchi}

We continue to decision problems about SD-tNBWs and SD-NBWs, and show that the exponential succinctness comes with a price: the complexity of all the problems we study coincides with the one known for tNBWs and NBWs. Accordingly, we only prove lower bounds for SD-tNBWs. Matching upper bounds follow from the known complexity for tNBWs, and same bounds for SD-NBWs follow from linear translations between SD-tNBWs and SD-NBWs.
The problems we study are {\em language containment:}  given two SD-tNBWs $\A_1$ and $\A_2$, decide whether $L(\A_1) \subseteq L(\A_2)$, {\em universality:} given an SD-tNBW $\A_1$, decide whether $L(\A_1)=\Sigma^\omega$, and {\em minimization: } given an SD-tNBW $\A_1$ and an integer $k \geq 1$, decide whether there is an SD-tNBW $\A_2$ such that $L(\A_1) = L(\A_2)$ and $|\A_2| \leq k$. 

The exponential succinctness of SD automata motivates also the study of the {\em D-to-SD minimization problem}. Here, we are given a tDBW $\A_1$ and an integer $k \geq 1$, and we need to decide whether there is an SD-tNBW $\A_2$ such that $L(\A_1) = L(\A_2)$ and $|\A_2| \leq k$. For automata on finite words, the D-to-N minimization problem is known to be PSPACE-complete~\cite{JR93}.

Note that a lower bound for universality implies a lower bound also for language containment. We still start with language containment, as it is much simpler.

\begin{theorem}\label{contain SD-tNBWs hard thm}
	The language-containment problem for SD-tNBWs is PSPACE-hard.	
\end{theorem}

\begin{proof}
	We describe a reduction from the universality problem for NFWs. Given an NFW $\N$ over $\Sigma$, let $\N'$ be an NFW over $\Sigma\cup \{\$\}$ such that $L(\N')= \$\cdot L(\N)\cdot \$$. 
	Now, let $\A_1$ be a 1-state tDBW over $\Sigma\cup \{\$\}$ such that $L(\A_1)= \infty \$$, and let $\A_2$ be the SD-tNBW obtained by applying the operation from Theorem \ref{NFW to SD-tNBW operation thm} on $\N'$. 
	Note that $L(\A_2) = \infty (\$\cdot L(\N) \cdot \$)$ and $|\A_2|=|\N|+3$.
	
	We claim that $\N$ is universal iff $L(\A_1) \subseteq L(\A_2)$. 
	First, if $L(\N)=\Sigma^*$, then $L(\A_2) = \infty (\$\cdot \Sigma^* \cdot \$ )= \infty \$$ and so $L(\A_1) \subseteq L(\A_2)$. Conversely, if there is a word $x\in \Sigma^*\setminus L(\N)$, then the word $w = (\$\cdot x)^\omega$ is in  $L(\A_1) \setminus L(\A_2)$.  Indeed, $w$ has infinitely many $\$$'s, yet for every infix $\$ \cdot y \cdot \$$ of $w$, we have that $y \not \in L(\N)$, and so $w \not \in L(\A_2)$. 
\end{proof}

The proof in Theorem~\ref{contain SD-tNBWs hard thm} uses $\infty \$$ as the ``contained language". For the universality problem, we cannot relay on hints from words in the contained language. In particular, taking the union of $\A_2$ there with an automaton for ``only finitely many $\$$'s" results in an automaton that is not SD, and it is not clear how to make it SD. Consequently, we have to work much harder. Specifically, we prove PSPACE hardness by a generic reduction from polynomial space Turing machines. Such reductions associate with a Turing machine $T$ an automaton $\A$ that recognizes the language $R$ %of words that do not encode legal accepting computations of $T$, and so $R=\Sigma^*$ iff the machine does not have an accepting computation. 
of words that do not encode legal rejecting computations of $T$, and so $R=\Sigma^*$ iff the machine has no rejecting computations.
The automaton $\A$ is nondeterministic, as it has to guess violations of attempts to encode legal accepting computations. In order to replace $\A$ by an SD automaton, we manipulate the Turing machine so that the language of the generated automaton is of the form $\infty R$, for which we can construct an SD-tNBW.

\begin{theorem}\label{universality buchi thm}
	The universality problem for SD-tNBWs is PSPACE-hard.	
\end{theorem}

\begin{proof}
	We do a reduction from polynomial-space Turing machines. Given
	a Turing machine $T$ with space complexity $s:\mathbb{N}\rightarrow\mathbb{N}$, we construct in time polynomial in $|T|$ and $s(0)$, an SD-tNBW $\A$ of size polynomial in $T$ and $s(0)$, such that $\A$ is universal iff $T$ accepts the empty tape\footnote{This is sufficient, as one can define a generic reduction from every language $L$ in PSPACE as follows. Let $T_L$ be a Turing machine that decides $L$ in polynomial space $f(n)$. On input $w$ for the reduction, the reduction considers the machine $T_w$ that on every input, first erases the tape, writes $w$ on its tape, and then runs as $T_L$ on $w$. Then, the reduction outputs an automaton $\A$, such that $T_w$ accepts the empty tape iff $\A$ is SD. Note that the space complexity of $T_w$ is $s(n) = \max(n ,f(|w|))$, and that $w$ is in $L$ iff $T_w$ accepts the empty tape. Since $\A$ is constructed in time polynomial in $s(0)=f(|w|)$ and $|T_w|=\mathrm{poly}(|w|)$, it follows that the reduction is polynomial in $|w|$.}. Let $n_0=s(0)$. Thus, each configuration in the computation of $T$ on the empty tape uses at most $n_0$ cells. 
	We  assume that $T$ halts from all configurations (that is, not just from these reachable from an initial configuration of $T$); Indeed, by adding a polynomial-space counter to $T$, one can transform a polynomial-space Turing machine that need not halt from all configurations to one that does halt. 
	We also assume, without loss of generality, that once $T$ reaches a final (accepting or rejecting) state, it erases the tape, moves with its reading head to the leftmost cell, and moves to the initial state. Thus, all computations of $T$ are infinite and after visiting a final configuration for the first time, they eventually consists of repeating the same finite computation on the empty tape that uses at most $n_0$ tape cells.
	
	We define $\A$ so that it accepts a word $w$ iff (C1) no suffix of $w$ is an encoding of a legal computation of $T$ 
	that uses at most $n_0$ tape cells, or (C2) $w$ has infinitely many infixes that encode the accepting state of $T$.
	
	It is not hard to see that $T$ accepts the empty tape iff $\A$ is universal. Indeed, if $T$ accepts the empty tape, and there is a word $w$ that  does not satisfy (C1), thus $w$ has a suffix that is an encoding of a legal computation of $T$ that uses at most $n_0$ cells, then the encoded computation eventually reaches a final configuration, from which it eventually repeats the accepting computation of $T$ on the empty tape infinitely many times, and so $w$ satisfies (C2). Conversely, If $T$ rejects the empty tape, then the word $w$ that encodes the computation of $T$ on the empty tape does not satisfy (C1) nor (C2), and so $\A$ does not accept $w$.

	Finally, the fact that $T$ is a polynomial-space Turing machine enables us to define $\A$ with polynomially many states, as we detail next. %In Appendix~\ref{app tm}. 
	Let $T=\zug{\Gamma,Q,\rightarrow,q_0,q_{acc},q_{rej}}$ be a Turing machine with polynomial space complexity $s: \mathbb{N} \rightarrow \mathbb{N}$, where $\Gamma$ is the tape-alphabet, $Q$ is the set of states, $\rightarrow \subseteq Q \times
	\Gamma \times Q \times \Gamma \times \{L,R\}$ is the transition
	relation (we use $(q,a) \rightarrow (q',b,\Delta)$ to indicate that
	when $T$ is in state $q$ and it reads the input $a$ in the current
	tape cell, it moves to state $q'$, writes $b$ in the current tape
	cell, and its reading head moves one cell to the left/right, according
	to $\Delta$), $q_0$ is the initial state, $q_{acc}$ is the accepting state, and $q_{rej}$ is the rejecting one. 
	%The transitions function $\rightarrow$ is defined also for the final states $q_{acc}$ and $q_{rej}$: when a computation of $T$ reaches them, it erases the tape, goes to the leftmost cell in the tape, %ofer1 #2and moves to the initial state $q_0$.
	%Recall that $s: \mathbb{N} \rightarrow \mathbb{N}$ is the polynomial space function of $T$. Thus, when $T$ runs on the empty tape, it uses at most $n_0=s(0)$ cells. 
	
	We encode a configuration of $T$ that uses at most $n_0 = s(0)$ tape cells by a word of the form 
	$\# \gamma_1 \gamma_2 \ldots (q,\gamma_i) \ldots \gamma_{n_0}$.
	That is, the encoding of a configuration starts with a special letter $\#$, and all its other letters
	are in $\Gamma$, except for one letter in $Q \times \Gamma$.
	The meaning of such a configuration is that the $j$'th cell in $T$, for
	$1 \leq j \leq n_0$, is $\gamma_j$-labeled, the head of $T$ points
	at $i$'th cell, and $T$ is in the state $q$. For example, the initial
	configuration of $T$ on the empty tape is encoded as
	$\# (q_0,b) b \ldots b$ (with $n_0-1$ occurrences of $b$'s) where $b$
	stands for an empty cell.
	We can now encode a computation of $T$ as a sequence of configurations.

	Let $\Sigma=\{\#\} \cup \Gamma \cup (Q \times \Gamma)$, and let
	$\# \sigma_1 \ldots \sigma_{n_0} \# \sigma'_1 \ldots \sigma'_{n_0}$
	be two successive configurations of $T$. We also set
	$\sigma_0$, $\sigma'_0$, and $\sigma_{n_0+1}$ to $\#$.
	For each triple $\zug{\sigma_{i-1},\sigma_i,\sigma_{i+1}}$ with $1
	\leq i \leq n_0$, we know, by the transition relation of $T$, what
	$\sigma'_i$ should be. In addition, the letter $\#$ should repeat
	exactly every $n_0+1$ letters.
	Let $next(\sigma_{i-1},\sigma_i,\sigma_{i+1})$ denote our
	expectation for $\sigma'_i$. That is,
	\begin{itemize}
		\item $next(\sigma_{i-1},\sigma_i,\sigma_{i+1})=\sigma_i$ if $\sigma_i=\#$, or if non of $\sigma_{i-1},\sigma_i$ and $\sigma_{i+1}$ are in $Q\times\Gamma$.

		\item
		$next((q,\gamma_{i-1}),\gamma_i,\gamma_{i+1})=
		next((q,\gamma_{i-1}),\gamma_i,\#)=$
		\begin{center}
			$\left\{
			\begin{array}{ll}
			\gamma_i & \mbox{If $(q,\gamma_{i-1}) \rightarrow
				(q',\gamma'_{i-1},L)$}\\
			(q',\gamma_i) & \mbox{If $(q,\gamma_{i-1}) \rightarrow
				(q',\gamma'_{i-1},R)$}
			\end{array}
			\right.$
		\end{center}
		\item
		$next(\gamma_{i-1},\gamma_i,(q,\gamma_{i+1}))=
		next(\#,\gamma_i,(q,\gamma_{i+1}))=$
		\begin{center}
			$\left\{
			\begin{array}{ll}
			\gamma_i & \mbox{If $(q,\gamma_{i+1}) \rightarrow
				(q',\gamma'_{i+1},R)$}\\
			(q',\gamma_i) & \mbox{If $(q,\gamma_{i+1}) \rightarrow
				(q',\gamma'_{i+1},L)$}
			\end{array}
			\right.$
		\end{center}
		\item
		$next(\gamma_{i-1},(q,\gamma_{i}),\gamma_{i+1})=
		next(\gamma_{i-1},(q,\gamma_{i}),\#)= 	next(\#,(q,\gamma_{i}),\gamma_{i+1})
		= \gamma'_i$ where $(q,\gamma_{i}) \rightarrow (q',\gamma'_{i},\Delta)$. \footnote{We note that in the case where the head of $T$ is at leaft most cell, and $T$ moves its head to the left should be handeled slightly differently, as we assume that the tape of $T$ is left bounded.}
		
	\end{itemize}
	
	\noindent
	Note that every word that encodes a legal computation of $T$ that uses at most $n_0$ tape cells must be consistent with $next$.

	Consider now an infinite word $w$. Note that if $w$ satisfies (C1), then $w$ has infinitely many infixes that witness a violation of the encoding of a legal computation of $T$ that uses at most $n_0$ tape cells. Accordingly, we can write the language of $\A$ as $\infty R$, where $R$ is the language of one-letter words of the form $(q_{acc},\gamma)$, %for the accepting state $q_{acc}$ of $T$,
	or words %of length $n_0+2$
	that describe a violation of a legal computation that uses at most $n_0$ cells.
	Words that describe a violation, either are words of length $n_0+1$ that describe a violation of the encoding of a single configuration, or are words of length $n_0+3$ that describe a violation $next$.
	
	We can now define an NFW $\N$ for $R$ of size polynomial in $T$ and $n_0$, and thus the required SD-tNBW $\A$ can be obtained from $\N$ by applying the construction described in  Theorem~\ref{NFW to SD-tNBW operation thm}. 
	The NFW $\N$ is defined as follows. 
	From its initial states, $\N$ guesses to check whether the input word is a single-letter word of the form $(q_{acc},\gamma)$, or guesses to check whether the input is a violation of the encoding of a legal computation of $T$ that uses at most $n_0$ tape cells: this amounts to guessing that the input is a word  of length $n_0+1$ that has no $\#$'s, or starts with $\#$, yet does not encode a legal configuration, or amounts to guessing that the input is a word of length $n_0 + 3$ that describes a violation of $next$, where $\N$ checks whether the input start with three letters $\sigma_{1},\sigma_2, \sigma_3$ describing three successive letters in the encoding
	and ends with a letter $\sigma'_2$, which comes $n_0+1$ letters after $\sigma_2$ in the encoding, yet is different than $next(\sigma_1, \sigma_2, \sigma_3)$.
\end{proof}

We continue to the minimization problem. Note that here, a PSPACE upper bound does not follow immediately from the known PSPACE upper bound for tNBWs, as the candidate automata need to be SD. Still, as SDness can be checked in PSPACE \cite{AKL21}, a PSPACE upper bound follows.  Also note that here, the case of SD-NBWs is easy, as a non-empty SD-NBW is universal iff it has an equivalent SD-NBW with one state. For transition-based acceptance, the language of a single-state SD-tNBW need not be trivial, and so we have to examine the specific language used for the universality PSPACE-hardness proof: % (see proof in Appendix~\ref{app min buchi is hard thm}): 

\begin{theorem}\label{min buchi is hard thm}
	The minimization problem for SD-tNBWs is PSPACE-hard.
\end{theorem}

\begin{proof}
	We argue that the SD-tNBW $\A$ in the proof of Theorem~\ref{universality buchi thm} has an equivalent one-state SD-tNBW iff $\A$ is universal. To see why, assume towards contradiction that $T$ rejects the empty tape and there is a one-state SD-tNBW $\D$ equivalent to $\A$. Recall that we encode computations of $T$ as a sequence of successive configurations of length $n_0$ each, separated by the letter $\#$.
	On the one hand, as the word $\#^\omega$ has infinitely many violations, then it is in $L(\A)$, and so the $\#$-labeled self-loop at the state of $\D$ must be an $\alpha$-transition. On the other hand, the word $w$ that encodes the computation of $T$ on the empty tape is not in $L(\A)$, yet it  has infinitely many $\#$'s.
\end{proof}

As we show below, the minimization problem stays hard even when we start from a deterministic automaton:

\begin{theorem}\label{tDBW to SD-tNBW thm}
	The D-to-SD minimization problem for B\"uchi automata is PSPACE-hard.
\end{theorem}

\begin{proof}
We start with B\"uchi automata with transition-based acceptance, and describe a reduction from the D-to-N minimization problem for automata on finite words: given a DFW $\A_1$, and an integer $k\geq 1$, decide whether there is an NFW $\A_2$ such that $L(\A_1) = L(\A_2)$ and $|\A_2| \leq k$. In \cite{JR93}, the authors prove that the problem is PSPACE-hard, in fact PSPACE-hard already for DFWs that recognize a language that has no good prefixes.
	 
	 Consider a language $R\subseteq \Sigma^*$. The reduction is based on a construction that turns an NFW for $ R$ into an SD-tNBW for $\infty (\$\cdot R \cdot \$)$, for a letter $\$\notin \Sigma$. Note that by applying the construction from Theorem~\ref{NFW to SD-tNBW operation thm} on the language $\$\cdot R \cdot \$ $, we can get an SD-tNBW for $\infty (\$\cdot R \cdot \$)$. The construction there, however, does not preserve determinism. Therefore, we need a modified polynomial construction, which takes advantage of the $\$$'s. We describe the modified construction below. 
	 
	 Given an NFW $\A = \zug{\Sigma, Q, Q_0, \delta, F}$, and a letter $\$\notin \Sigma$, we construct the tNBW $\A' = \zug{\Sigma \cup \{\$\}, Q, Q_0, \delta', \alpha}$, where for all states $q\in Q$ and letters $\sigma\in \Sigma$, we have that $\delta'(q, \sigma) = \delta(q, \sigma)$. Also, $\delta'(q, \$) = Q_0$, and $\alpha = \{\zug{s, \$, q}: q\in Q_0 \text{ and } s\in F\}$. Thus, $\A'$ is obtained from $\A$ by adding $\$$-transitions from all states to $Q_0$. A new transition is in $\alpha$ iff its source is an accepting state of $\A$. 
	 %
	 %\color{red}
	 %
	 It is not hard to see that the construction of $\A'$ preserves determinism. Indeed, if $\A$ is deterministic, then all the $\$$-transitions in $\A'$ are deterministic since they are directed to the single initial state of $\A$. 
	 To show that $\A'$ is SD and $L(\A') = \infty (\$\cdot L(\A) \cdot \$)$, we prove below, similarly to Theorem~\ref{NFW to SD-tNBW operation thm}, that for all states $q\in Q$, it holds that $L(\A'^q) = \infty (\$\cdot L(\A) \cdot \$)$.

	 Consider a state $q\in Q$, we first show that $L(\A'^q) \subseteq (\$\cdot L(\A) \cdot \$)$. Consider a word $w\in L(\A'^q)$ and let $r=r_0, r_1, r_2, \ldots$ be an accepting run of $\A'^q$ on $w$. We show that for every $i\geq 1$, there are $ k_1, k_2$ such that $i\leq k_1\leq k_2$ and $w[k_1 , k_2+1]$ is an infix of $w$ in $\$\cdot L(\A) \cdot \$$. Given $i\geq 1$, let $k_2$ be such that $t = \zug{r_{k_2}, w[k_2+1], r_{k_2+1}}$ is an $\alpha$-transition, but not the first $\$$-transition that $r$ reads when it runs on the suffix $w[i, \infty]$, and let $k_1$ be the maximal index such that $i\leq k_1\leq k_2$ and $\zug{r_{k_1-1}, w[k_1], r_{k_1}} = \zug{r_{k_1-1}, \$, r_{k_1}}$. Since $t$ is not the first $\$$-transition that $r$ traverses when it reads the suffix $w[i, \infty]$, then $k_1$ exists. Note that the fact that new transitions lead to $Q_0$ in $\A'$, and the maximality of $k_1$, imply that the run $r_{k_1}, r_{k_1+1}, \ldots, r_{k_2}$ is a run of $\A$ on the infix $w[k_1+1, k_2]$. Also, as $t$ is an $\alpha$ transition, the latter run is an accepting run of $\A$. Now, as the transitions $\zug{r_{k_1-1}, w[k_1], r_{k_1}} $ and $t$ are $\$$-transitions, we get that $r_{k_1-1}, r_{k_1}, \ldots, r_{k_2}, r_{k_2+1}$ is a run on a word in $\$\cdot L(\A) \cdot \$$. Thus, $w[k_1, k_2+1]$ is an infix of $w$ in $\$\cdot L(\A) \cdot \$$. 
	 
	 Next, we show that $\infty (\$\cdot L(\A) \cdot \$) \subseteq L(\A'^q)$. Consider a word $w\in \infty (\$\cdot L(\A) \cdot \$) $, and consider a run of $\A'^q$ on $w$ that upon reading an infix $\$\cdot x \cdot \$$ in  $\$\cdot L(\A^{q_0}) \cdot \$$, for some $q_0\in Q_0$, moves to $q_0$ while reading $\$$, and then follows an accepting run of $\A^{q_0}$ on $x$, and finally moves to a state in $Q_0$, traversing $\alpha$ while reading $\$$. 
	 Note that the transitions in $\A'$ enable such a behavior. Thus, such a run traverses $\alpha$ infinitely often and is thus accepting.

	 We now describe the reduction. Given a DFW $\A$ over $\Sigma$ such that $L(\A)$ has no good prefixes, and given an integer $k$, the reduction returns the polynomial tDBW $\A'$ and the integer $k$. We prove next that the reduction is correct. Thus, the DFW $\A$ has an equivalent NFW with at most $k$ states iff the tDBW $\A'$ has an equivalent SD-tNBW with at most $k$ states. For the first direction, if $\B$ is an NFW equivalent to $\A$ whose size is at most $k$, then by the above construction, it holds that $\B'$ is an SD-tNBW whose size is at most $k$, and $L(\B') = \infty (\$ \cdot L(\B) \cdot \$) = \infty (\$ \cdot L(\A) \cdot \$) = L(\A')$.  
	 
	 Conversely, if $\B'$ is an SD-tNBW for $\infty (\$ \cdot L(\A) \cdot \$)$ whose size is at most $k$, then as $L(\A)$ has no good prefixes, we get by Theorem~\ref{SD-tNBW to NFW operation thm} that there is an NFW $\B$ for $L(\A)$ whose size is at most $k$, and we are done.
	 
%Finally, since the transitions between automata with transition-based and state-based acceptance may involve a linear blow-up, here we need to be careful in extending the result to the state-based setting. In Appendix~\ref{app DBW to SD-NBW hard}, we prove that the arguments in Theorem~\ref{tDBW to SD-tNBW thm} can be adapted to automata with state-based acceptance, thus PSPACE-completeness holds also for B\"uchi automata with state-based acceptance.

Finally, we show next that despite the fact that transitions between automata with transition-based and state-based acceptance may involve a linear blow-up, it is possible to adapt the above arguments to automata with state-based acceptance, thus PSPACE-completeness holds also for B\"uchi automata with state-based acceptance.  Essentially, given a DFW $\A$ and an integer $k$, the reduction returns a DBW for $\infty (\$ \cdot L(\A) \cdot \$)$ and the integer $k+1$.

To see why such a reduction exists,
note that the construction of the SD-tNBW $\A'$ from the NFW $\A$ can be adapted to the state-based setting by adding a single new $\alpha$ state to $\A'$ that behaves as the initial states of $\A$, and captures infixes in $\$\cdot L(\A) \cdot \$$. 
Formally,  given an NFW $\A = \zug{\Sigma, Q, Q_0, \delta, F}$, and a letter $\$\notin \Sigma$, we construct the NBW $\A' = \zug{\Sigma \cup \{\$\}, Q \cup \{q_{acc}\}, Q_0, \delta', \{q_{acc}\}}$, where for all states $q\in Q$ and letters $\sigma\in \Sigma$, we have that $\delta'(q, \sigma) = \delta(q, \sigma)$. Also, $\delta'(Q\setminus F, \$) = Q_0$, $\delta'(F, \$) = \{q_{acc}\}$, $\delta(q_{acc}, \$) = Q_0$, and $\delta'(q_{acc}, \sigma) = \delta(Q_0, \sigma)$. Intuitively, the state $q_{acc}$ mimics traversal of $\alpha$ transitions.

Then, Theorem~\ref{SD-tNBW to NFW operation thm} can also be adapted to SD-NBWs as follows. If $\A$ is an SD-NBW for $\infty (\$\cdot R 
\cdot \$)$ for some nontrivial language $R\subseteq \Sigma^*$ that has no good prefixes, then one can construct an NFW $\N$ for $R$ with at most $|\A| - 1$ states. 
Essentially, we view $\A$ as a transition-based automaton $\A_t$ by considering  transitions that leave $\alpha$ states as $\alpha$ transitions. Then, we consider the proof of Theorem~\ref{SD-tNBW to NFW operation thm} when applied on $\A_t$.  
As $R$ has no good prefixes, then the state $q_{acc}$ of $\N$ is unreachable. Hence, all runs from $Q^S_0$ cannot reach $\alpha$ transitions in $\A_t$. In particular, $\alpha$ states of $\A$ are unreachable in $\N$ and they can be removed with $q_{acc}$ as well.
\end{proof}

\section{Semantically Deterministic co-\buchi Automata}

In this section we study SD co-B\"uchi automata. Here too, our results are based on constructions that involve encodings of NFWs by SD-tNCWs. Here, however, the constructions are more complicated, and we first need some definitions and notations.

Consider a tNCW  $\A = \langle \Sigma, Q, Q_0, \delta, \alpha\rangle$. We refer to the SCCs we get by removing $\A$'s $\alpha$-transitions as the \emph{$\overline{\alpha}$-components} of $\A$; that is, the \emph{$\overline{\alpha}$-components} of $\A$ are the SCCs of the graph $G_{\A^{\bar{\alpha}}} = \langle Q, E^{\bar{\alpha}} \rangle$, where $\zug{q, q'}\in E^{\bar{\alpha}}$ iff there is a letter $\sigma\in \Sigma$ such that $\langle q, \sigma, q'\rangle \in \overline{\alpha}$. 
We say that $\A$ is \emph{normal} if
there are no $\overline{\alpha}$-transitions connecting different $\overline{\alpha}$-components. That is,
for all states $q$ and $s$ of $\A$, if there is a path of $\overline{\alpha}$-transitions from $q$ to $s$, then there is also a path of $\overline{\alpha}$-transitions from $s$ to $q$.
Note that an accepting run of $\A$ eventually gets trapped in one of $\A$'s $\overline{\alpha}$-components. In particular, accepting runs in $\A$ traverse transitions that leave $\overline{ \alpha}$-components only finitely often. Hence, we can add transitions among $\overline{\alpha}$-components to $\alpha$ without changing the language of the automaton. Accordingly, in the sequel we assume that given tNCWs are normal.

We proceed to encoding NFWs by SD-tNCWs. For a language $R\subseteq \Sigma^*$, we define the language $\bowtie_\$\!\!(R) \subseteq (\Sigma \cup \{\$\})^\omega$, by  $$\bowtie_\$\!\!(R) = \{w: \text{if $w$ has infinitely many $\$$, then it has a suffix in $(\$ R)^\omega$}\}.$$ 

Also,  we say that a finite word $x\in \Sigma^*$ is a {\em bad infix} for $R$ if for all words $w\in \Sigma^*$ that have $x$ as an infix, it holds that $w\in \overline{R}$.

 Note that for every language $R \subseteq \Sigma^*$, we have that $\bowtie_\$\!\!(R)  = \overline{\infty (\$\cdot \overline{R} \cdot \$)}$. Thus, $\bowtie_\$\!\!(R)$ complements $\infty (\$\cdot \overline{R} \cdot \$)$. Yet, unlike tDCWs and tDBWs, which dualize each other, SD-tNBWs and SD-tNCWs are not dual. Hence, 
adjusting Theorems \ref{NFW to SD-tNBW operation thm} and \ref{SD-tNBW to NFW operation thm} to the co-\buchi setting, requires different, in fact more complicated, constructions. 

\begin{theorem}\label{NFW to SD-tNCW operation thm}
	Given an NFW $\N$, one can obtain, in linear time, an SD-tNCW $\A$ such that \mbox{$L(\A)=\bowtie_\$\!\!(L(\N))$} and $|\A|=|\N|$.
\end{theorem}

\begin{proof}
	Given $\N  = \zug{\Sigma, Q, Q_0, \delta, F}$, we obtain $\A$ by adding $\$$-transitions from all states to $Q_0$. The new transitions are in $\alpha$ iff they leave a state in $Q\setminus F$. Formally, $\A= \zug{\Sigma \cup \{\$\}, Q, Q_0, \delta', \alpha}$, where for all $s \in Q$ and $\sigma \in \Sigma$, we have that $\delta'(s,\sigma)=\delta(s,\sigma)$, and $\delta'(s,\$) = Q_0$.  Then, $\alpha =  \{\zug{s, \$, q}: q\in Q_0 \text{ and } s \in Q\setminus F\}$.
	It is easy to see that $|\A|=|\N|$. In order to prove that $\A$ is SD and $L(\A)=\bowtie_\$\!\!(L(\N))$, we prove next %in Appendix~\ref{app NFW to SD-tNCW operation thm}
	that for every state $q\in Q$, it holds that $L(\A^q)=\bowtie_\$\!\!( L(\N))$. 
	%Essentially, this follows from the fact that $\A$ can avoid traversing $\alpha$ when it reads an input of the form $x\cdot \$$, only when $x\in L(\N)$.
	
	First, we show that $\bowtie_\$\!\!(L(\N))\subseteq L(\A^q)$. Consider a word $w\in \bowtie_\$\!\!(L(\N))$. If $w$ has finitely many $\$$'s, then as $\alpha$ transitions in $\A$ are $\$$-labeled, we have that all runs of $\A^q$ on $w$ traverse $\alpha$ finitely often and thus are accepting. Assume that $w = z \cdot y$, where $y$ is of the form  $\$ x_1 \$ x_2\$x_3 \cdots$, where for all $i$, $x_i\in L(\N)$. An accepting run $r$ of $\A^q$ on $w$ can be defined arbitrarily upon reading $z$. Then, upon reading an infix of the form $\$x_i$, $r$ moves to $q_0$ upon reading $\$$, where $q_0\in Q_0$ and $x_i\in L(\N^{q_0})$. Then, $r$ proceeds by following an accepting run of $\N^{q_0}$ on $x_i$. As accepting runs of $\N$ end in states in $F$, and $\$$-transitions from $F$ are in $\overline{ \alpha}$, we get that $r$ does not traverse $\alpha$ upon reading the $\$$'s after $x_1$. Also, as transitions of $\N$ are $\overline{\alpha}$ transitions in $\A$, we get that $r$ follows a run that does not traverse $\alpha$ when it reads $x_i$. Hence, $r$ eventually vists only $\overline{ \alpha}$ transitions and thus is accepting.
	
	Next, we show that $ L(\A^q)\subseteq \bowtie_\$\!\!(L(\N))$. Consider a word $w = \sigma_1\cdot \sigma_2 \cdots\in L(\A^q)$ and let $r=r_0, r_1, \ldots $ be an accepting run of $\A^q$  on $w$. If $w$ has finitely many $\$$'s, then $w\in \bowtie_\$\!\!(L(\N))$ and we are done. Assume that $w$ has infinitely many $\$$'s and
	let $t$ be such that $r[t-1, \infty] = r_{t-1}, r_t, r_{t+1}, \ldots$ is a run that does not traverse $\alpha$ on $w[t, \infty]$. To conclude the proof, we show that between every two consecutive $\$$'s in $w[t, \infty]$ there is a word in $L(\N)$. That is, we  show that for all  $t\leq t_1 < t_2$  such that $\sigma_{t_1} = \sigma_{t_2} = \$$ and $w[t_1+1, t_2-1]$ has no $\$$'s,  it holds that $w[t_1+1, t_2-1]\in L(\N)$.  
	As $\$$-transitions lead to $Q_0$, we have that $r_{t_1}, r_{t_2} \in   Q_0$.  Recall that $r[t-1, \infty]$ does not traverse $\alpha$, in particular, 
	we have that $r[t_1, t_2] = r_{t_1}, r_{t_1+1}, \ldots r_{t_2}$ is a run from $Q_0$ to $Q_0$ that traverses only $\overline{\alpha}$ transitions,  on the infix   $w[t_1+1, t_2]$. Also, $r[t_1, t_2]$ traverses $\$$ only in its last transition. As the transition $\zug{r_{t_2-1}, \sigma_{t_2}, r_{t_2}} $ is an $\overline{ \alpha}$ $\$$-transition that leads to $Q_0$,  then $r_{t_2-1}\in F$. Also, since all the transitions that $r[t_1, t_2-1]$ traverses are $\Sigma$-labeled, then they are transitions in $\N$, and so $r[t_1, t_2-1]$ is an accepting run of $\N$ on $w[t_1+1, t_2 - 1]$. Hence $w[t_1+1, t_2 - 1] \in L(\N)$ and we are done.
\end{proof}

\begin{theorem}\label{SD-tNCW to NFW operation thm}
	Consider a language $R \subseteq \Sigma^*$ and a letter $\$\notin \Sigma$. For every tNCW $\A$ such that \mbox{$L(\A)=\bowtie_\$\!\!(R)$}, there exists an NFW $\N$ such that $L(\N) = R$ and $|\N|\leq |\A| + 1$. In addition, if $R$ has bad infixes, then $|\N| \leq |\A|$. 
\end{theorem}

\begin{proof}
	Let $\A = \zug{ \Sigma \cup \{\$\}, Q, Q_0, \delta, \alpha}$ be a tNCW for \mbox{$\bowtie_\$\!\!(R)$}. We assume that $\A$ is normal and all of its states are reachable.
	We define the NFW $\N = \zug{\Sigma, Q\cup \{q_{rej}\}, Q^{\N}_0, \delta_\N, F_\N}$, where 
	$Q^{\N}_0 = \{q\in Q:  \text{there is a state $q'$ such that $\zug{q', \$, q}\in \overline{ \alpha}$}\}$, 	$F_{\N} = \{q\in Q:  \text{there is a state $q'$ such that $\zug{q, \$, q'}\in \overline{ \alpha}$}\}$,
	%
	%
	%
	%$Q^{\N}_0 = \{q\in Q: \text{there is an $\overline{ \alpha}$-transition in $\A$ of the form $q' \xrightarrow{\$} q$}\}$, $F_\N = \{q\in Q: \text{there is an $\overline{ \alpha}$-transition in $\A$ of}$ $\text{the form $q \xrightarrow{\$} q'$}\}$,
	and for every two states $q, s \in Q$ and letter $\sigma\in \Sigma$, it holds that $s\in \delta_\N(q, \sigma)$ iff $\zug{q, \sigma, s} \in \overline{\alpha}$.
	Also, if $\{q\}\times \{\sigma\} \times \delta(q, \sigma) \subseteq \alpha$, then $\delta_\N(q, \sigma) = q_{rej}$. Also, for all letters $\sigma \in \Sigma$, it holds that $\delta_\N(q_{rej}, \sigma) = q_{rej}$; that is, $q_{rej}$ is a rejecting sink.
	Thus, $\N$ tries to accept words $x\in \Sigma^*$ for which there is a run in $\A$ that does not traverse $\alpha$ on the word $\$\cdot x \cdot \$$. 
	%Clearly, $|\N| = |\A|$. 
	
	We prove that $L(\N) = R$. 
	We first prove that $R\subseteq L(\N)$. Consider a word $x\in R$, and let $r = r_0, r_1, r_2, \ldots$ be an accepting run of $\A$ on $(\$\cdot x)^\omega$. As $r$ is accepting, there are $i < j$ such that $r_i, r_{i+1}, \ldots, r_j$ is a run that does not traverse $\alpha$ on $\$\cdot x \cdot \$$. 
	By the definition of $\delta_\N$, $r_{i+1}\in Q^{\N}_0, r_{j-1}\in F_\N$, and so $r_{i+1}, r_{i+2}, \ldots, r_{j-1}$ is an accepting run of $\N$ on $x$.

	%Then, the definition of $\delta_\N$ implies that the run $r[i+1, j-1] = r_{i+1}, r_{i+2}, \ldots, r_{j-1}$ exists in $\N$, and since $r[i, j]$ strats and ends by traversing safe $\$$-labeled transitions, then the definition of $\A$ implies that $r_{i+1}\in Q^{\N}_0$ and $r_{j-1}\in F_\N$, and thus $r[i+1, j-1]$ is an accepting run of $\N$ on $x$. 
	
	We prove next that $L(\N)\subseteq R$. 
	Consider a word $x = \sigma_1\cdot \sigma_2 \cdots \sigma_n\in L(\N)$ and let $r = r_0, r_1, \ldots r_n$ be an accepting run of $\N$ on $x$. 
	As $r$ ends an accepting state of $\N$, then it does not visit $q_{rej}$. Hence,
	the definition of $\delta_\N$ implies that $r$ is a run that does not traverse $\alpha$ in $\A$. Also, as $r_0\in Q^{\N}_0$ and $r_n\in F_\N$, there are states $q_1, q_2\in Q$ such that $\zug{q_1, \$, r_0}, \zug{r_n, \$, q_2}\in \overline{ \alpha}$. Hence, $q_1, r_0, r_1, \ldots, r_n, q_2$ is a run that does not traverse $\alpha$ in $\A$ on the word $\$\cdot x \cdot \$$. As $\A$ is normal, 
	there is a word $y\in (\Sigma \cup \{\$\})^*$ such that there is a run that does not traverse $\alpha$ of $\A^{q_2}$ on $y$ that reaches $q_1$. Therefore, $(\$\cdot x \cdot \$ \cdot y)^\omega\in L(\A^{q_1})$.
	%the latter run can be extended to a safe cycle from $q_1$ to $q_1$ upon reading the word $\$\cdot x \cdot \$ \cdot y$ for some $y\in (\Sigma \cup \{\$\})^*$, in particular $(\$\cdot x \cdot \$ \cdot y)^\omega\in L(\A^{q_1})$. 
	As all the states of $\A$ are reachable, it follows that there is a word $z\in  (\Sigma \cup \{\$\})^*$ such that $q_1\in \delta(Q_0, z)$, and thus $z\cdot (\$\cdot x \cdot \$ \cdot y)^\omega \in L(\A)$. Hence, as \mbox{$L(\A)=\bowtie_\$\!\!(R)$}, we get that $x\in R$. %, and we are done.
	
	Since the state space of $\N$ is $Q\cup \{q_{rej}\}$, we have that $|\N| = |\A|+1$. Moreover, if $R$ has a bad infix $x$, then $x\in \overline{R}$. 
	Hence, if we consider the word $x^\omega$, then as it has no $\$$'s, it is accepted by $\A$. Hence, there is a state $q\in Q$ such that $\A^q$ has a run on $x^\omega$ that does not leave the $\overline{ \alpha}$-component of $q$. As we detail below, %in Appendix~\ref{app SD-tNCW to NFW operation thm}, 
	the fact that $x$ is a bad infix of $R$ implies that the $\overline{\alpha}$-component of $q$ does not contain $\$$-transitions. Hence, since the only states in $Q$ that are reachable in $\N$ lie in $\overline{ \alpha}$-components that contain a $\$$-transition, we 
	can remove the state $q_{rej}$ from $\N$ and make $q$ a rejecting sink instead. Hence, in this case, we have that $|\N| \leq |\A|$, and we are done.
	
	Let $x$ be a bad infix of $R$, and let $q\in Q$ be such that $\A^q$ has a run on $x^\omega$ that does not leave the $\overline{ \alpha}$-component of $q$. Let $S(q)$ denote the $\overline{ \alpha}$-component of $q$. 
	To conclude the proof, we show that $S(q)$ does not contain a $\$$-transition. 
	Assume by contradiction that $S(q)$ contains a $\$$-transition $s_1 \xrightarrow{\$} s_2$. As $q$ has a run on $x^\omega$ that does not leave $S(q)$, we get that $q$ has a finite run of the form $q \xrightarrow{x} q_x$ on $x$ that does not leave $S(q)$. As $\A$ is normal, we get that there is a run in $S(q)$ that starts at $q_x$, passes through the transition  $s_1 \xrightarrow{\$} s_2$, and finally ends in $q$: $q_x \xrightarrow{y_1} s_1 \xrightarrow{\$} s_2 \xrightarrow{y_2} q$. Hence, we get that there is a cycle in $S(q)$ of the form $q\xrightarrow{x\cdot y} q$ on a word of the form $x\cdot y$, where $y$ contains a $\$$. Hence, $(x\cdot y)^\omega$ is accepted from $q$ in $\A$. 

	Note that the word $(x\cdot y)^\omega = x\cdot y \cdot x \cdot y \cdots $ has infinitely many infixes of the form $\$\cdot u \cdot \$$, where $u\in \Sigma^*$ is a bad infix of $R$. To see why, consider the infix $z = y\cdot x \cdot y$ that appears infinitely often in $(x\cdot y)^\omega$, and consider the infix $\$ \cdot u \cdot \$$ of $z = y\cdot x \cdot y$ that starts at last $\$$ of the first $y$ and ends at the first $\$$ of the second $y$. Clearly, such $u$ is over $\Sigma$, and since it contains $x$ as an infix, then it is also a bad infix. In particular, $u\in \overline{ R}$.
	As all the states of $\A$ are reachable, there is a word $v\in (\Sigma\cup \{\$\})^*$ such that $q\in \delta(q_0, v)$. Hence, the word $v\cdot (x\cdot y)^\omega$ is accepted in $\A$ even though it has infinitely many infixes of the form $\$\cdot u \cdot \$$, where $u\in \overline{ R}$, and we have reached a contradiction to the fact that $L(\A) = \bowtie_\$\!\!( R)$. 
\end{proof}

\subsection{Succinctness and Complementation}\label{SD-tNCW succinct sec}

In this section we study the succinctness of SD co-B\"uchi automata with respect to deterministic ones, and the blow-up involved in their complementation. Recall that unlike B\"uchi automata, for co-B\"uchi automata, HD automata are exponentially more succinct than deterministic ones. Accordingly, our results about the succinctness of SD automata with respect to HD ones are more surprising than those in the setting of B\"uchi automata. 

\begin{theorem}
	\label{SD-tNCW succinct thm}
	There is a family $L_1,L_2,L_3,\ldots$ of languages such that for every $n \geq 1$, there is an SD-tNCW with $3n+3$ states that recognizes $L_n$, yet every tDCW or HD-tNCW that recognizes $L_n$ needs at least $2^n$ states.
\end{theorem}

\begin{proof}
	For $n \geq 1$, consider the alphabet $\Sigma_n = [n] \cup \{\#\}$, and the language $R_n  = \{x\cdot \# \cdot i: x\in [n]^+$ and $i$ appears in $x \} \subseteq \Sigma^*_n$. We define $L_n = \bowtie_\$\!\!(R_n)$.
	Recall that a word over $z \in (\Sigma_n \cup \{\$\})^*$ is good if $z = \$\cdot  x  \cdot \#  \cdot i$, where $x\in [n]^+$ and $i$ appears in $x$, and note that $L_n$ consists of exactly the words with finitely many $\$$'s, or that have a suffix that is a concatenation of good words. 
	First, it is not hard to see that $R_n$ can be recognized by an NFW $\N_n$ with $3n+3$ states, for example, a candidate for $\N_n$ can be obtained from the NFW in Figure~\ref{N_n} by removing the $\$$-transitions from $q_0$, and letting $q_0$ guess to behave as any state in $\bigcup_{i\in [n]} \{q^1_i\}$. By Theorem~\ref{NFW to SD-tNCW operation thm}, there is an SD-tNCW for $L_n$ with $3n+3$ states.
	
	%In order to show that an HD-tNCW for $L_n$ needs at least $2^n$ states, we rely on properties of HD-tNCWs \cite{KS15,AK19} and argue that (1) an HD-tNCW for $L_n$ has a state $q$ such that $L_{\overline{\alpha}} (q) = \{w: \text{there is a run from $q$ on $w$ that does not traverse $\alpha$}\}$ is such that $(\$\cdot R_n)^\omega \subseteq L_{\overline{\alpha}} (q)$, and (2) the $\overline{ \alpha}$-component of $q$ is of size at least $2^n$.
	%
	%As we detail below, %in Appendix~\ref{app cobuchi},
	%the first claim follows from the fact that HD-tNCWs can be assumed to be {\em $\overline{\alpha}$ deterministic}, thus all their $\bar{\alpha}$-transitions are deterministic \cite{KS15, AK19}.  
	%The second claim follows from the fact that
	%the $\overline{ \alpha}$-component of $q$ should detect good subwords, and should remember for this subsets of $[n]$. 

	%%%%%
	%%%%%

	We show next that an HD-tNCW for $L_n$ needs at least $2^n$ states. Consider an HD-tNCW $\A = \zug{\Sigma_n \cup \{\$\}, Q, q_0, \delta, \alpha}$ for $L_n$. By \cite{KS15,AK19}, we can assume that $\A$ is {\em $\overline{\alpha}$-deterministic}, thus for all states $q$ and letters $\sigma$, we have that there is at most  one state $q'$ such that $\zug{q, \sigma, q'}\in \overline{ \alpha}$. We also assume that $\A$ is normal, and all its states are reachable.  To prove that $|Q|\geq 2^n$,
	we first prove that there is a state $q\in Q$ such that $(\$\cdot R_n)^\omega \subseteq L_{\overline{\alpha}} (q)$, where $L_{\overline{\alpha}} (q) = \{w: \text{there is a run from $q$ on $w$ that does not traverse $\alpha$}\}$. % and $(\$\cdot \overline{R_n} \cdot \$ \cdot (\Sigma_n \cup \{\$\})^\omega) \cap L_{\overline{\alpha}} (q) = \emptyset$. 
	Then, we argue that the $\overline{\alpha}$-component of $q$ has at least $2^n$ states.

	%We define a tDCW $\D$ such that $L(\A) \subseteq L(\D)$. Then, using the structure of $\D$ and the fact that $\A$ is $\overline{ \alpha}$ deterministic, we conclude that there is a state $q\in Q$ such that $(\$\cdot R_n)^\omega \subseteq L_{\overline{\alpha}} (q)$. 
	For a nonempty set $S\in 2^Q$, and a letter $\sigma \in \Sigma_n \cup \{\$\}$, define $\delta^{\overline{\alpha}} (S, \sigma) = \{q'\in Q: \text{there is a state $q\in S$} $ such that $\zug{q, \sigma, q'}\in \overline{ \alpha}\}$.
	Consider the tDCW $\D = \zug{\Sigma_n \cup \{\$\}, 2^Q\setminus \emptyset, Q, \delta_\D, \alpha_\D}$, where for every nonempty set $S\in 2^Q\setminus \emptyset$, and letter $\sigma \in \Sigma_n \cup \{\$\}$, it holds that $\delta_\D(S, \sigma) =\delta^{\overline{\alpha}} (S, \sigma)$, if $\delta^{\overline{\alpha}} (S, \sigma)\neq \emptyset$, and $\delta_\D(S, \sigma)  = Q$, otherwise. The former type of transitions are in $\overline{\alpha_\D}$, and the latter are in $\alpha_\D$. Thus, the state space of $\D$ encodes nonempty sets $S\in 2^Q$ of states, and $\D$ starts a run from the set of all states $Q$. Reading its input, $\D$ keeps track of runs in $\A$ that have not traversed $\alpha$ from the last traversal of $\alpha_\D$, and when all tracked runs in $\A$ have traversed $\alpha$, the tDCW $\D$ resets again to the set of all states $Q$. 
	
	We prove next that\footnote{In fact, using the fact that all the states of $\A$ are reachable, the fact that HD-tNCWs can be assumed to be SD \cite{KS15}, and the fact that $\A$ recognizes $L_n$, one can prove that $\D$ is equivalent to $\A$.} $L(\A) \subseteq L(\D)$. Consider a word $w\in L(\A) $, and let $r = r_0, r_1, r_2, \ldots$ be an accepting run of $\A$ on $w$. As $r$ is accepting, there is $i\geq 0$ such that $r[i, \infty] = r_i, r_{i+1}, r_{i+2}, \ldots$  does not traverse $\alpha$ on the suffix $w[i+1, \infty]$. Let $r_S = S_0, S_1, S_2, \ldots $ denote the run of $\D$ on $w$, and assume towards contradiction that $r_S$ traverses $\alpha_\D$ infinitely often. Then, consider an index $j\geq i$ such that $\zug{S_j, w[j+1], S_{j+1}} \in \alpha_\D$. By the definition of $\D$, we have that $S_{j+1} = Q$. Consider the run $r[j+1, \infty]$ on the suffix $w[j+2, \infty]$. As $r_{j+1}\in S_{j+1}$, $r[j+1, \infty]$ is a run of $r_{j+1}$ that does not traverse $\alpha$, and $\D$ keeps track of runs in $\A$ that do not traverse $\alpha$ from the last reset, then it must be the case that the run $r_S[j+1, \infty] = S_{j+1}, S_{j+2}, S_{j+3}, \ldots$ does not traverse $\alpha_\D$, and we have reached a contradiction.

	Next, we show that there is a state $q\in Q$ such that $(\$\cdot R_n)^\omega \subseteq L_{\overline{\alpha}} (q)$. Assume towards contradiction that for all states $q\in Q$, it holds that there is a word $w_q\in (\$\cdot R_n)^\omega \setminus L_{\overline{\alpha}} (q)$.
	%Recall that $L(\A) = L_n \subseteq L(\D)$.
	We define a sequence $S_0, S_1, S_2, \ldots $ of states in $\D$, and a sequence $w_1, w_2, w_3, \ldots$ of finite words in $(\$\cdot R_n)^+$ such that for all $i \geq 0$, it holds that $\delta_\D({S_i}, w_{i+1}) = S_{i+1}$. 
	We take $S_0 = Q$.  For all $i\geq 0$, given $S_i$,  we define $w_{i+1}$ and $S_{i+1}$ as follows. 
	Consider a state $q_i \in S_i$ and the word $w_{q_i}\in (\$\cdot R_n)^\omega \setminus L_{\overline{\alpha}} (q_i)$. Then, as we argue below, there exists a prefix $h_q$ of $w_{q_i}$ in $(\$\cdot R_n)^+$ such that all runs of $q_i$ on $h_q$ traverse $\alpha$. 
	We take $w_{i+1}$ to be $h_q$, and take $S_{i+1} = \delta_\D(S_i, w_{i+1})$. We explain now why such a prefix exists. Assume contrarily that there are indices $k_1 < k_2 < k_3 < \cdots$ such that for all $j\geq 1$, the word $w_{q_i}[1, k_j] = \sigma_1\cdot \sigma_2 \cdots \sigma_{k_j}\in (\$\cdot R_n)^+$ is a prefix of $w_{q_i} = \sigma_1 \cdot \sigma_2 \cdot \sigma_3 \cdots $ such that $q_i$ has a run $r_{k_j}$ that does not traverse $\alpha$ on $w_{q_i}[1, k_j]$. Then, as $\A$ is $\overline{ \alpha}$ deterministic, we get that the run $r_{k_{j_2}}$ extends the run $r_{k_{j_1}}$ for all $j_2 > j_1$. Hence, there is an infinite run of $q_i$ that does not traverse $\alpha$ on $w_{q_i}$, and we have reached a contradiction to the fact that $w_{q_i} \notin L_{\overline{\alpha}} (q_i)$. 
	
	Consider sequences $S_0, S_1, S_2, \ldots$ and $w_1, w_2, w_3, \ldots$ as above. First, it holds that the word $w = w_1\cdot w_2 \cdot w_3 \cdots $ is in $(\$\cdot R_n)^\omega \subseteq L_n$, and, by definition, $r_S = S_0 \xrightarrow{w_1} S_1 \xrightarrow{w_2} S_2 \cdots$ is a run of $\D$ on $w$.  Recall that $L(\A) = L_n \subseteq L(\D)$; in particular $w\in L(\D)$. 
	To reach a contradiction, we prove below that the run $r_S$ of $\D$ on $w$, is rejecting. We claim that it is sufficient to show that for all $i\geq 0$, either the run $r_{i \to i+1} = {S_i}\xrightarrow{w_{i+1}} S_{i+1}$ traverses $\alpha_\D$, or $|S_i| > |S_{i+1}|$. Indeed, if the latter holds, and there is $i\geq 0$ such that the run $S_i \xrightarrow{w_{i+1}} S_{i+1} \xrightarrow{w_{i+2}}  S_{i+2}\cdots $ does not traverse $\alpha_\D$, then we get that $|S_j| > |S_{j+1}|$ for all $j\geq i$, which is a impossible.

	So we prove next that for all $i\geq 0$, either the run $r_{i \to i+1} = {S_i}\xrightarrow{w_{i+1}} S_{i+1}$ traverses $\alpha_\D$, or $|S_i| > |S_{i+1}|$.  
	Consider a state $S$ of $\D$, a letter $\sigma \in \Sigma_n \cup \{\$\}$, and note the following.
	If for all $q \in S$ it holds that $(\{q\} \times  \{\sigma\} \times \delta(q, \sigma) )\cap \overline{\alpha} \neq \emptyset$, 
	then as $S$ proceeds to the $\overline{\alpha}$ $\sigma$ successors of the states in $S$ upon reading $\sigma$, and as $\A$ is $\overline{\alpha}$ deterministic, we get that $|S| \geq |\delta_\D(S, \sigma)|$ and $\zug{S, \sigma, \delta_\D(S, \sigma)} \in \overline{\alpha_\D}$. Otherwise, if there is a state $q\in S$ such that $(\{q\} \times \{\sigma\} \times \delta(q, \sigma)) \cap \overline{\alpha} =\emptyset$, then either $\zug{S, \sigma, \delta_\D(S, \sigma)} \in \alpha_\D$, or $|S| > |\delta_D(S, \sigma)|$. 
	Consider $i\geq 0$. Let $w_{i+1} = \sigma_1\cdot \sigma_2 \cdots \sigma_{n_i}$, and consider the run 
	$r_{i \to i+1} = {S_i}\xrightarrow{w_{i+1}} S_{i+1} = T_0, T_1, \ldots, T_{n_i}$.
	If, by contradiction, the run $r_{i \to i+1}$ does not traverse $\alpha_\D$ and $|S_i| = |S_{i+1}|$, then by the above,  for all $0 \leq j < n_i$ and $q\in T_j$, we have that $( \{q\} \times \{\sigma_{j+1}\} \times \delta(q, \sigma_{j+1})) \cap \overline{\alpha} \neq \emptyset$. Hence, if we consider a state $t_0\in T_0$, we get that there is a transition  $\zug{t_0, \sigma_1, t_1}\in \overline{ \alpha}$, and by the defition of $\delta_\D$, we have that $t_1\in T_1$. Then, proceeding iteratively from $t_1$, we get that there is a run of $t_0$ in $\A$ that does not traverse $\alpha$ on the word $w_{i+1} = \sigma_1 \cdot \sigma_2 \cdots \sigma_{n_i}$. As the latter holds for all $t_0\in T_0 = S_i$; in particular,  for $t_0 = q_i$, we get a contradiction to the definition of $w_{i+1}$.

	What we have shown so far is that there exists a state $q\in Q$ such that $(\$\cdot R_n)^\omega \subseteq L_{\overline{\alpha}} (q)$. To complete the proof, we prove next that the $\overline{\alpha}$-component of $q$ has at least $2^n$ states.
	Assume towards contradiction that the $\overline{ \alpha}$-component of $q$ has strictly less than $2^n$ states. For every  subset $S\subseteq [n]$, 
	we abuse notation and use $S$ to also denote a word in $[n]^*$ consisting of exactly the numbers that appear in $S$.
	%let $z_S\in [n]^*$ be a word consisting of exactly the numbers that appear in $S$. 
	Then, consider the word $\$\cdot S$. The word $\$\cdot S$ can be extended to a word in $(\$\cdot R_n)^\omega \subseteq L_{\overline{\alpha}} (q)$. Hence, $q$ has a run $r_S$ that does not traverse $\alpha$ on the word $\$\cdot S$, and as $\A$ is normal, we have that $r_S$ ends in a state that belongs to the $\overline{ \alpha}$-component of $q$. As the $\overline{ \alpha}$-component of $q$ has strictly less than $2^n$ states, there are sets $S, T \subseteq [n]$ and $i\in [n]$ such that $i\in S\setminus T$, and $r_S$ and $r_T$ end in the same state $q'$. 
	As $\A$ is $\overline{ \alpha}$ deterministic, we get that the run $r_S$ is the only run that does not traverse $\alpha$ from the state $q$ on the word $\$ \cdot S$. Hence, as $\$\cdot S \cdot \#\cdot  i \cdot \$$ can be extended to a word in $(\$\cdot R_n)^\omega \subseteq L_{\overline{\alpha}} (q) $ and $r_S$ ends in $q'$, we get that there is a run of the form $q' \xrightarrow{\#\cdot i\cdot \$} q''$ that traverses only $\overline{ \alpha}$ transitions, in particular, it does not leave the $\overline{\alpha}$-component of $q$. Now consider the run $r = q \xrightarrow{\$\cdot T} q' \xrightarrow{\#\cdot i \cdot \$} q''$. The run $r$ is a run that does not traverse $\alpha$ from the state $q$ on the word $\$\cdot T\cdot \#\cdot i \cdot \$$.
	As $\A$ is normal, there is a word $y\in (\Sigma_n \cup \{\$\})^*$ such that there is a run that does traverse $\alpha$ of $\A^{q''}$ on $y$ that reaches $q$. Therefore, $(\$\cdot T\cdot \#\cdot i \cdot \$ \cdot y)^\omega \in  L_{\overline{\alpha}} (q) \subseteq L(\A^q)$. As all the states of $\A$ are reachable, there is a word $z\in (\Sigma_n \cup \{\$\})^*$ such that $q\in \delta(q_0, z)$. Hence, the word $z\cdot (\$\cdot T\cdot \#\cdot  i \cdot \$ \cdot y)^\omega$ is accepted by $\A$ even though it has infinitely many infixes in $\$ \cdot \overline{ R_n} \cdot \$$, and we have reached a contradiction.
\end{proof}

\begin{remark}\label{alt min co proof remark}
	{\rm As has been the case with \buchi automata, here too, an analysis of the application of the HD-tNCW minimization algorithm on a tDCW for $L_n$ leads to a slightly tighter bound. %Specifically, in Appendix~\ref{app alt min co proof remark}, we describe a tDCW for $L_n$ with $2^{n+1} + 1$ states, and in the full version we prove that the application of the HD-tNCW minimization algorithm on it does not make it smaller. 
	Specifically, in Appendix~\ref{app alt min co proof remark}, we describe a tDCW for $L_n$ with $2^{n+1} + 1$ states, such that the application of the HD-tNCW minimization algorithm on it does not make  it smaller. 		
	}\hfill \qed
\end{remark}

%of big stam

\begin{theorem}\label{SD-tNCW comp thm}
There is a family $L_1,L_2,L_3,\ldots$ of languages such that for every $n \geq 1$, there is an SD-tNCW with $O(n)$ states that recognizes $L_n$, yet every SD-tNBW that recognizes $\overline{L_n}$ needs at least $2^{O(n)}$ states.
\end{theorem}

\begin{proof}
Let $\Sigma=\{0,1\}$, and let $R_n=\{w: w \in (0+1)^* \cdot(0 \cdot (0+1)^{n-1} \cdot 1 +  1 \cdot (0+1)^{n-1} \cdot 0) \cdot (0+1)^*\}$. Thus, $R_n$ is the language of all words that contain two letters that are at distance $n$ and are different. 
We define $L_n=\bowtie_\$\!\!(R_n)$.

It is easy to see that $R_n$ can be recognized by an NFW with $O(n)$ states. Hence, by Theorem~\ref{NFW to SD-tNCW operation thm}, there is an SD-tNCW with $O(n)$ states for $L_n$.
We prove next that every SD-tNBW for $\overline{L_n}$ needs at least $2^{O(n)}$ states. 
For this, note that $\overline{L_n}$ consists of all words $w$ such that $w$ contains infinitely many $\$$'s yet has no suffix in $(\$ R_n)^\omega$. Thus, $w$ contains infinitely many infixes in $\$ \overline{R_n} \$$. Therefore, $\overline{L_n} = \infty (\$ \cdot \overline{R_n} \cdot \$)$. 
It is not hard to prove that an NFW for $\overline{R_n}$ needs at least $2^{O(n)}$ states. Then, by Theorem~\ref{SD-tNBW to NFW operation thm}, this bound is carried over to a $2^{O(n)}$ lower bound on an SD-tNBW for $\overline{L_n}$. 
\end{proof}

\subsection{Decision Problems}

We continue to decision problems for SD-tNCWs and SD-NCWs. As in Section~\ref{dp buchi}, we state only the lower bounds, and for SD-tNCWs. Matching upper bounds and similar results for SD-NCWs follow from the known upper bounds for tNCWs and the known linear translations between state-based and transition-based automata. 

We start with language-containment and universality. 

\begin{theorem}\label{universality SD-tNCWs thm}
The language-containment and universality problems for SD-tNCWs are PSPACE-hard.
\end{theorem}

\begin{proof} 
We describe a reduction from universality of NFWs to universality of SD-tNCWs. Given an NFW $\N$ over $\Sigma$, the reduction returns the SD-tNCW$ \A$ over $\Sigma\cup \{\$\}$  that is obtained by applying the operation from Theorem~\ref{NFW to SD-tNCW operation thm} on $\N$. Thus, $L(\A) = \bowtie_\$\!\!(L(\N))$.

First, if $L(\N)=\Sigma^*$, then $(\Sigma\cup \{\$\})^*\cdot (\$ L(\N))^\omega = \infty \$$, and so $L(\A)=(\Sigma\cup \{\$\})^\omega$. Conversely, if there is a word $x\in \Sigma^*\setminus L(\N)$, then the word $w = (\$ x)^\omega$ is not in $L(\A)$.  Indeed, $w$ has infinitely many $\$$'s, yet it has a word not in $L(\N)$ between every two consecutive $\$$'s. 
\end{proof}

We continue to the minimization problem.
Note that while minimization is PSPACE-complete for NCWs, it is NP-complete for DCWs and in PTIME for HD-tNCWs. Thus, our PSPACE lower bound suggests a significant difference between SD and HD automata. 
Note also that, as has been the case with \buchi automata, the case of SD-NCWs is easy, as a non-empty SD-NCW is universal iff it has an equivalent SD-NCW with one state, which is not the case for SD-tNCWs: %(see proof in Appendix~\ref{app min co is hard them}):

\begin{theorem}\label{min co is hard thm}
The minimization problem for SD-tNCWs is PSPACE-hard.
\end{theorem}

\begin{proof}
	We analyze the reduction from Theorem~\ref{universality SD-tNCWs thm} and show that an NFW $\N$ over $\Sigma$ is universal iff the SD-tNCW $\A$ over $\Sigma\cup \{\$\}$ that is constructed in Theorem~\ref{NFW to SD-tNCW operation thm} for $\bowtie_\$\!\!(L(\N))$ has an equivalent SD-tNCW with a single state. 
	
	First, note that $\N$ is universal iff $\bowtie_\$\!\!(L(\N))$ is universal. Hence, if $\N$ is universal, then $\A$ has an equivalent SD-tNCW with a single state.  
	%\color{red}
	For the other direction, assume that $\A$ has an equivalent SD-tNCW $\D$ with a single state, and assume towards contradiction that $\N$ is not universal. Then, there is a word $x \notin L(\N)$. Consider the word $(\$ x)^\omega$. By the definition of $\bowtie_\$\!\!(L(\N))$, we have that $(\$ x)^\omega \not \in L(\A)$. Hence, $(\$ x)^\omega$ is not accepted by $\D$ and so at least one of the letters $\sigma$ in the word $\$ x$ is such that the $\sigma$-labeled self-loop at the state of $\D$ is an $\alpha$-transition. %, and so $\D$ rejects the  word $(\$ \sigma)^\omega$.
	Now as $L(\A) = \bowtie_\$\!\!(L(\N))$, we get that $\D$ accepts all words that have no $\$$'s; in particular, the  $\alpha$-self-loop at the state of $\D$ must be $\$$-labeled.
	
	Recall that the universality problem for NFWs is PSPACE-hard already for NFWs $\N$ that accept $\epsilon$. For example, this is valid for the NFW that is generated in the generic reduction from a PSPACE Turing machine, and which  accepts all words that do not encode a valid computation of the Turing machine. Hence, for such an NFW $\N$, the word $\$^\omega$ is in $\bowtie_\$\!\!(L(\N))$, and should be accepted by $\D$, and we have reached a contradiciton.
\end{proof}

We continue to the D-to-SD minimization problem, showing it stays PSPACE-hard.

\begin{theorem}\label{tDCW to SD-tNCW thm}
	The D-to-SD minimization problem for co-B\"uchi automata is PSPACE-hard.
\end{theorem}

\begin{proof}
	We reduce from the D-to-N minimization problem for automata on finite words, which is already PSPACE-hard for languages that have bad infixes \cite{JR93}. We start with co-B\"uchi automata with transition-based acceptance. 
	
	Consider the construction from Theorem~\ref{NFW to SD-tNCW operation thm}. Recall it takes an NFW $\A$ as input, returns an SD-tNCW $\A'$ of the same size for $\mbox{$\bowtie_\$\!\!(L(\A))$}$, and preserves determinism. 
	
	Consider a DFW $\A$ over $\Sigma$ such that $L(\A)$ has a bad infix, and an integer $k$. The reduction returns the SD-tNCW $\A'$ constructed from $\A$ in Theorem~\ref{NFW to SD-tNCW operation thm}, and the integer $k$. 
	Recall that $L(\A')=\mbox{$\bowtie_\$\!\!(L(\A))$}$, and that the construction  in Theorem~\ref{NFW to SD-tNCW operation thm} preserves determinism. Thus, the automaton $\A'$ is really a tDCW. %Also, when $L(\A)$ has a bad infix, then $\A'$ is of the same size as $\A$.  

	We prove next that the reduction is correct. That is, the DFW $\A$ has an equivalent NFW with at most $k$ states iff the tDCW $\A'$ has an equivalent SD-tNCW with at most $k$ states.
	For the first direction, if $\B$ is an NFW equivalent to $\A$ whose size is at most $k$, then, by applying to it the construction from Theorem~\ref{NFW to SD-tNCW operation thm}, we get an SD-tNCW $\B'$ whose size is at most $k$, and $L(\B') = \mbox{$\bowtie_\$\!\!(L(\B))$} = \mbox{$\bowtie_\$\!\!(L(\A))$} = L(\A')$.  
	
	Conversely, if $\B'$ is an SD-tNCW for $\mbox{$\bowtie_\$\!\!(L(\A))$}$ whose size is at most $k$, then as $L(\A)$ has bad infixes, we get by Theorem~\ref{SD-tNCW to NFW operation thm} that there is an NFW $\B$ for $L(\A)$ whose size is at most $k$, and we are done.
	
%In Appendix~\ref{app DCW to SD-NCW hard}, we extend the proof to co-B\"uchi automata with state-based acceptance. 
Finally, as in \buchi automata, it is not hard to see that the above arguments can be adapted to state-based automata, and so the D-to-SD minimization problem is PSPACE-hard also for co-\buchi automata with state-based acceptance. Essentially, given a DFW $\A$ and an integer $k$, the reduction returns a DCW for $\mbox{$\bowtie_\$\!\!(L(\A))$}$ and the integer $k+1$.

The above reduction exists since the construction from Theorem~\ref{NFW to SD-tNCW operation thm}
can be adapted to state-based automata, essentially, by letting the co-\buchi automaton have one additional $\alpha$ state that behaves as the initial states, and that is reachable from every state in $Q\setminus F$ upon reading $\$$.

Also, Theorem~\ref{SD-tNCW to NFW operation thm} can be adapted to the state-based setting. Indeed, if $\A$ is an NCW for $\mbox{$\bowtie_\$\!\!(R)$}$ for some language language $R\subseteq \Sigma^*$ that has a bad infix $x$, then one can construct an NFW $\N$ for $R$ with at most $|\A| - 1$ states. 
To see why, we view $\A$ as a transition-based automaton $\A_t$ by considering  transitions that touch $\alpha$ states as $\alpha$ transitions. Then, we consider the proof of Theorem~\ref{SD-tNCW to NFW operation thm} when applied on $\A_t$.  
We claim that all states $p$ that are $\alpha$ states of $\A$ can be removed from $\N$ without affecting its language. Indeed, the $\overline{\alpha}$-component of such a state $p$ in $\A_t$ is $\{p\}$ and contains no transitions, in paritcular, no $\$$-transitions. 
Also, note that such states $p$ are distinct from the state $q$ that has a run that does not traverse $\alpha$ transitions on $x^\omega$.
Thus, they can be removed from $\N$ along with the state $q_{rej}$.  Hence, we get in this case that $|\N|\leq |\A| - 1$.
\end{proof}

\section{Semantically Deterministic Weak Automata}
\label{weak sec}

By \cite{AKL21}, SD-NWWs need not be DBP or even HD. For completeness we describe here the example from~\cite{AKL21}, as it highlights the challenges in SD-NWW determinization. Consider the automaton $\A$
in Figure~\ref{weak fig}.  

\begin{figure}[htb]	
	\begin{center}
		\includegraphics[height=2.9cm]{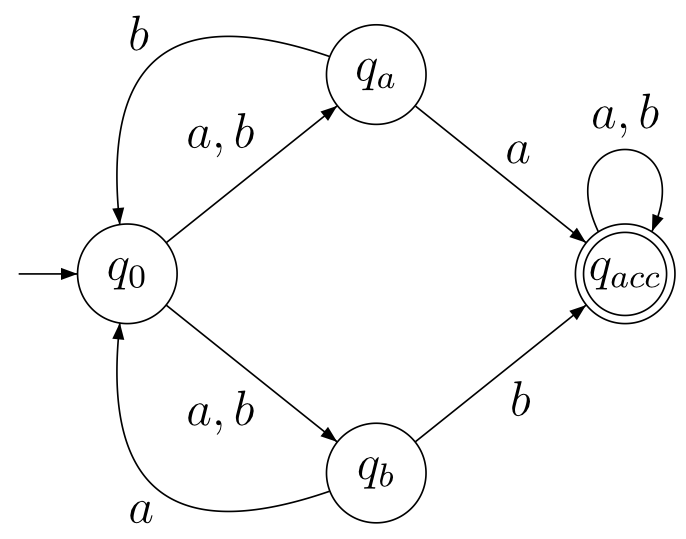}
		\caption{An SD-NWW that is not DBP.}
		\label{weak fig}
	\end{center}
\end{figure}

It is easy to see that $\A$ is weak, and all its states are universal,
and so it is SD. On the other hand, $\A$ is not HD as every strategy has a word with which it does not reach $q_{acc}$ -- a word that forces every visit in $q_a$ and $q_b$ to be followed by a visit in $q_0$. Below we show that despite not being DBP, SD-NWWs can be determinized in polynomial time. Essentially, our proof is based on redirecting transitions of the SD-NWW to deep components in the automaton. In our example, note that while the SD-NWW $\A$ is not HD, it has a deterministic state $q_{acc}$ that recognizes $L(\A)$.

Consider an SD-NWW $\A= \zug{\Sigma, Q, q_0, \delta, \alpha}$. We denote the set of $\A$'s SCCs by $\C(\A)$, and the SCC containing a state $q$ by $C(q)$. Let $C_1 \leq C_2 \leq \cdots \leq C_m$ be a total order on the SCCs of $\A$, extending the partial order induced by $\delta$. That is, if $q' \in \delta(q,\sigma)$, for some letter $\sigma$, then $C(q) \leq C(q')$. When $C \leq C'$, we say that $C'$ is {\em deeper than\/} $C$. Thus, states along runs of a weak automaton proceed from SCCs to deeper ones, and eventually get stuck in some SCC.

If we had an algorithm that checks language equivalence between states in $\A$ in polynomial time, we could have a polynomial determinization algorithm that defines the $\sigma$-successor of a state as the deepest state among all states equivalent to its $\sigma$-successors (since $\A$ is SD, all these successors agree on their language). Since we still do not have such an algorithm (in fact, it would follow from our construction), we approximate language equivalence by an equivalence that follows from the semantic determinism of $\A$.  

For two states $s_1, s_2 \in Q$, we say that $s_1$ and $s_2$ are {\em $\delta$-close}, if there is a state $q \in Q$ and a word $w \in \Sigma^*$ such that $s_1,s_2 \in \delta(q,w)$. Note that the $\delta$-close relation refines equivalence, yet the converse does not hold. Indeed, the SDness property implies that $s_1 \sim_\A s_2$ for all $\delta$-close states $s_1$ and $s_2$. 
%orna1: no room for this
%\badernew{We exploit the latter fact, and show that restricting attention to $\delta$-equivalent states is sufficient to obtain an equivalent determinsitic automaton. We start by computing the $\delta$-equivalence relation:}

\begin{lemma}
	\label{deltaeqns}
	If $s_1$ and $s_2$ are $\delta$-close, then there is a state $q$ and word $w$ of length at most $|Q|^2$ such that $s_1,s_2 \in \delta(q,w)$.
\end{lemma}

\begin{proof}
	Consider $\delta$-close states $s_1$ and $s_2$. Let $q$ and $w$ be such that $s_1,s_2 \in \delta(q,w)$. If $|w| > |Q|^2$, then in the runs from $q$ to $s_1$ and from $q$ to $s_2$ there must be a pair of states $\zug{p_1,p_2}$ that repeats twice, and we can shorten $w$ by removing the subword between the two occurrences of $\zug{p_1,p_2}$. Formally, there is a partition of $w=x \cdot y \cdot z$ such that
$p_1,p_2 \in \delta(q,x)$, $p_1 \in \delta(p_1,y)$, $p_2 \in \delta(p_2,y)$, $s_1 \in \delta(p_1,z)$, and $s_2 \in \delta(p_2,z)$. Then, $s_1,s_2 \in \delta(q,x \cdot z)$. This process can repeat until $w$ is of length at most $|Q|^2$.
\end{proof}

In order to calculate the $\delta$-close relation in polynomial time, we define a sequence $H_0,H_1,H_2, \ldots \subseteq Q \times Q$ of relations, where $\zug{s_1,s_2} \in H_i$ iff there is a state $q$ and word $w$ of length at most $i$ such that  $s_1,s_2 \in \delta(q,w)$.  By Lemma~\ref{deltaeqns}, we are guaranteed to reach a fixed point after at most $|Q|^2$ iterations. The relations $H_i$ are defined as follows. 

\begin{itemize}
	\item
	$H_0 = \{\zug{q,q} : q \in Q\}$.
	\item
	For $i\geq 0$, we define $H_{i+1}=H_i \cup \{\zug{s_1,s_2} :$ there is  $\zug{q_1,q_2} \in H_i$ and letter $\sigma \in \Sigma$ such that
	$s_1 \in \delta(q_1,\sigma)$ and $s_2 \in \delta(q_2,\sigma)\}$.
\end{itemize}

Let $j \geq 0$ be such that $H_{j+1}=H_j$. It is not hard to see that $H_j$ is the $\delta$-close relation.
While the $\delta$-close relation is reflexive and symmetric, it is not transitive.
Now, let $H \subseteq Q \times Q$ be the closure of $H_j$ under transitivity.
That is, $\zug{s_1,s_2} \in H$ iff there is $k \geq 2$ and states $q_1,q_2,\ldots,q_k$ such that $q_1=s_1$, $q_k=s_2$ and
$\zug{q_i,q_{i+1}} \in H_j$ for all $1 \leq i < k$.  
%orna1: no room for this (moved some up)
%\badernew{
%	It is not hard to see that $H_j$ is the $\delta$-equivalence relation, in particular, it refines the $\sim_\A$ relation. Hence, the transitivity of $\sim_\A$ implies that $H$ refines $\sim_\A$ too.
	%
	%
	%
%	Also, observe that $H$ is an equivalence relation that can be obtained in polynomial time: $H_j$ is reflexive and symmetric, in particular, if we consider the simple undirected graph $T=\zug{Q, E}$, where $\{q, s\} \in E$ iff $q\neq s$ and $H_j(q, s)$, then we get that for all $s_1, s_2\in Q$, it holds that $\zug{s_1,s_2} \in H$ iff $s_1$ and $s_2$ lie in the same SCC of $T$.}

The following lemma implies that $H$ propagates to successor states.
\begin{lemma}
	\label{H propa lemma}
	If $H(s,s')$, then for every letter $\sigma \in \Sigma$ and states $q \in \delta(s,\sigma)$ and $q' \in \delta(s',\sigma)$, we have that $H(q,q')$.
\end{lemma}

\begin{proof}
		Consider the fixed point relation $H_j$. By definition, if $H_j(s,s')$, then for every letter $\sigma \in \Sigma$ and states $q \in \delta(s,\sigma)$ and $q' \in \delta(s',\sigma)$, we have that $H_{j+1}(q,q')$, and since $H_j$ is the fixed point, then $ H_j(q,q')$. Then, the lemma follows from the definition of $H$ as the closure of $H_j$ under transitivity.
	%
	%\color{red} 
	Indeed, if $H(s, s')$, then there is $k\geq 2$ and states $q_1, q_2, \ldots, q_k$ such that $q_1 = s, q_k = s'$ and $H_j(q_{i}, q_{i+1})$ for all $1\leq i < k$. Now if $s_1, s_2, \ldots, s_k$ are such that $s_i \in \delta(q_i, \sigma)$, then by the above considerations, we have that $H_j(s_i, s_{i+1})$ for all $1\leq i < k$. Hence, the definition of $H$ implies that $H(s_1, s_k)$, and we are done.
\end{proof}

It is easy to see that $H$ is an equivalence relation. 
Let ${\cal P}=\{P_1,\ldots,P_k\}$ be the set of the equivalence classes of $H$.
For each equivalence class $P \in {\cal P}$, we fix the {\em representative of\/} $P$ as some state in $P \cap C$, where $C \in \C(\A)$ is the deepest SCC that intersects $P$. Let $p_1,\ldots,p_k$ be the representatives of the sets in ${\cal P}$.
For a state $q \in Q$, let $\widetilde{q}$ denote the representative of the set $P \in {\cal P}$ with $q \in P$.
Note that as $H$ refines $\sim_\A$, we have that $q\sim_\A \widetilde{q}$.

\begin{theorem}
\label{sd nww det}
	Given an SD-NWW $\A$ with state space $Q$, we can construct, in polynomial time,
	an equivalent DWW $\D$ with state space $Q'$, for $Q' \subseteq Q$.
\end{theorem}

\begin{proof}
	
	Given $\A=\zug{\Sigma,Q,q_0,\delta,\alpha}$, we define $\D=\zug{\Sigma,Q',q'_0,\delta',\alpha \cap Q'}$, where
	\begin{itemize}
		\item
		$Q'=\{p_1,\ldots,p_k\}$. Note that indeed $Q' \subseteq Q$.
		\item
		$q'_0=\widetilde{q_0}$. 
		\item
		For $p \in Q'$ and $\sigma \in \Sigma$, we define $\delta'(p,\sigma)=\tilde{q}$, for some $q \in \delta(p,\sigma)$.
		Note that all the states in $\delta(p,\sigma)$ are $\delta$-close, and thus  belong to the same set in ${\cal P}$. Hence, the choice of $q$ is not important, as $\delta'(p,\sigma)$ is the representative of this set.
		%Note that SDness guarantees that all the states in $\delta(p,\sigma)$ belong to the same set in ${\cal P}$. Hence, the choice of $q$ is not important, as $\delta'(p,\sigma)$ is the representative of the set that contains all the states in $\delta(p,\sigma)$.
	\end{itemize}

	We prove that $\D$ is a DWW equivalent to $\A$.
First, in order to see that $\D$ is weak, consider states $p,p'$ such that $p' \in \delta'(p,\sigma)$. Thus, $p'$ is the representative of a state $q \in \delta(p,\sigma)$. As $\A$ is weak, we have that $C(p) \leq C(q)$. As $p'=\tilde{q}$, we have that $C(q) \leq C(p')$. Hence $C(p) \leq C(p')$, and we are done.

	We continue and prove that $L(\A) = L(\D)$. We first prove that $L(\A) \subseteq L(\D)$. Consider a word $w$ and assume that $w \in L(\A)$. Let $r=q_0,q_1,q_2,\ldots$ be an accepting run of $\A$ on $w$, and let $r'=s_0,s_1,s_2,\ldots$ be the run of $\D$ on $w$. If $r'$ is accepting, we are done. Otherwise, namely if $r'$ is rejecting, we point to a contradiction.
	Let $j \geq 0$ be the index in which $r'$ visits only states in $sinf(r')$. Note that all the states $s_j,s_{j+1},s_{j+2},\ldots$ belong to some SCC $C$ of $\A$. Since we assume that $r'$ is rejecting, and acceptance in $\D$ is inherited from $\A$, we get that $C$ is a rejecting SCC of $\A$. We claim that all the runs of $\A^{s_j}$ on the suffix $w[j+1,\infty]$ of $w$ are stuck in $C$. Thus, $w[j+1,\infty] \not \in L(\A^{s_j})$.  On the other hand, we claim that for all $i \geq 0$, we have that $q_i \sim_\A s_i$. Then, however, $w[j+1,\infty] \not \in L(\A^{q_j})$, contradicting the fact that $r$ is accepting:
	
	The easy part is to prove that for all $i \geq 0$, we have that $q_i \sim_\A s_i$. Indeed, it follows from the fact that for all $i \geq 0$, we have that $H(q_i,s_i)$ and the fact that $H$ refines $\sim_\A$. The proof proceeds by an induction on $i$. First, as for every state $q \in Q$, we have that $H(q,\widetilde{q})$, then $H(q_0,q'_0)$, and so the claim follows from $s_0$ being $q'_0$. Assume now that $H(q_i,s_i)$. Recall that $s_{i+1}$ is $\widetilde{q}$ for some $q \in \delta(s_i, w[i+1])$. Also, $q_{i+1} \in \delta(q_i, w[i+1])$. Then, by Lemma~\ref{H propa lemma}, we have that $H(q_{i+1}, q)$.  Since, in addition, we have that $H(q, \widetilde{q})$, then the transitivity of $H$ implies that $H(q_{i+1},s_{i+1})$, and we are done. 
	
	We proceed to the other part, namely, proving that  
	all the runs of $\A^{s_j}$ on the suffix $w[j+1,\infty]$ of $w$ are stuck in $C$.
    Let $u_{j+1},u_{j+2},\ldots$ be such that for all $i \geq 1$, it holds that $u_{j+i} \in \delta(s_{j+i - 1}, w[j+i])$ and $H(u_{j+i}, s_{j+i})$. Note that such states exist, as $s_{j+i-1} \xrightarrow{w[j+i]} s_{j+i}$ is a transition of $\D$ and so $s_{j+i} = \widetilde{q}$ for all $q\in \delta(s_{j+i - 1}, w[j+i])$.
	Consider a run $r''=v_0,v_1,v_2,\ldots$ of $\A^{s_j}$ on $w[j+1,\infty]$. Note that $v_0=s_j$, and so $H(v_0,s_j)$.
	Therefore, by Lemma~\ref{H propa lemma}, we have that $H(v_1,u_{j+1})$, implying $H(v_1,s_{j+1})$.
	By repeated applications of Lemma~\ref{H propa lemma}, we get that $H(v_i,u_{j+i})$, implying $H(v_i,s_{j+i})$, for all $i \geq 1$. 
	Assume now by way of contradiction that the run $r''$ leaves the SCC $C$, and so there is $i \geq 1$ such that $C(v_i) \not \leq C(s_{j+i})$. 
	As $H(v_i,s_{j+i})$, the states $v_i$ and $s_{j+i}$ are in the same equivalence class $P\in {\cal P}$. Then, the definition of $Q'$ implies that $C(s_{j+i}) \geq C(v_i)$, and we have reached a contradiction.

	It is left to prove that $L(\D) \subseteq L(\A)$. The proof follows similar considations and is detailed next. %in Appendix~\ref{app second dir dww}.
	Consider a word $w \in L(\D)$. Let $r'=s_0,s_1,s_2,\ldots$ be the run of $\D$ on $w$, and let $j \geq 0$ be the index in which $r'$ visits only states in $sinf(r')$; in particular, as argued above, $sinf(r')$ is included in some SCC $C$ of $\A$. Thus, all the states $s_j,s_{j+1},s_{j+2},\ldots$ are in $C$. Since $w\in L(\D)$, then $r'$ is accepting, and so $C$ is an accepting SCC of $\A$.
	As in the previous direction, it holds that $s_j$ is $\A$-equivalent to all the states in $\delta(q_0, w[1, j])$. Hence, to conclude that $w\in L(\A)$, we show that  there is an accepting run of $\A^{s_j}$ on $w[j+1,\infty]$.
	Assume by way of contradiction that all the runs of $\A^{s_j}$ on $w[j+1,\infty]$ are rejecting. In particular, all these runs leave the SCC $C$. As in the previous direction, this contradicts the choice of the states $s_j,s_{j+1},s_{j+2},\ldots$ as representatives of equivalence classes in ${\cal P}$.
\end{proof}

Since DWWs can be complemented by dualization (that is, by switching $\alpha$ and $\bar{\alpha}$), Theorem~\ref{sd nww det} implies the following.

\begin{theorem}
\label{sd nww comp}
Given an SD-NWW $\A$ with $n$ states, we can construct, in polynomial time, an SD-NWW (in fact, a DWW) that complements $\A$. 
\end{theorem}

Since the language-containment problem for DWW can be solved in NLOGSPACE, and minimization for DWW is similar to minimization of DFWs and can be solved in polynomial time \cite{Lod01}, we have the following.

\begin{theorem}
\label{sd nww dp}
The language-containment, universality, and minimality problems for SD-NWWs can be solved in polynomial time.
\end{theorem}

%\section{Discussion}
\stam{
The study of bounded forms of nondeterminism has attracted a lot of research in the last decade. Both since it is theoretically interesting and intriguing, and since models with bounded nondeterminism turn out to be useful in synthesis and reasoning about probabilistic systems. Many questions about different types of bounded nondeterminism are still open. In particular we do not understand the effect of the bounds on different acceptance conditions (not only Buchi vs. co-Buchi, but also state-based vs. transition-based conditions; for example, while minimization of GFG Buchi automata or deterministic co-Buchi automata is NP-complete, it is in PTIME for GFG co-Buchi automata with transition-based acceptance condition).

Semantic determinism (SD) seems very restrictive (it is "almost determinism" -- why would branches to states with the same language help? and indeed for automata on finite words they do not). Our results show that SD is in fact strong, and that unlike other bounds on nondeterminism, it affects the different acceptance conditions in the same way (except for weak automata). We found it very interesting, and hope that it will shed light on the problems that are still open.
}

\bibliography{full}
\appendix

\section{Missing Details}

\subsection{Missing details in Remark~\ref{alt min proof remark}}
\label{app alt min proof remark}

Consider the tDBW $\D_n  =\zug{\Sigma_n, 2^{[n]}\times \{a, c\}, \zug{\emptyset, c}, \delta, \alpha}$, defined as follows (see Figure~\ref{Dn}).  
The state space of $\D_n$  consists of two copies of subsets of $[n]$: the $a$-copy  $2^{[n]}\times \{a\}$, with ``a" standing for ``accumulate", and the $c$-copy $2^{[n]}\times \{c\}$, with ``c" standing for ``check".
Essentially, the state $\zug{S, a}$ in the $a$-copy remembers that we have read only numbers in $[n]$ after the last $\$$ and also remembers the set $S\subseteq [n]$ of these numbers.  The state $\zug{S, c}$ in the $c$-copy remmbers, in addition, that we have just read $\#$. 
Accordingly, $\zug{S, c}$  checks whether the next letter is in $S$ in order to traverse $\alpha$.
Effectively, $\D_n$ traverses $\alpha$ when a good infix is detected.
Formally,  for every state $\zug{S, o} \in  2^{[n]}\times \{a, c\}$, and letter $\sigma \in \Sigma_n$, we define:

$$
\delta(\zug{S,a}, \sigma)=
\begin{cases}
\zug{S \cup \{\sigma\}, a},  &  \text{if $\sigma \in [n]$} \\
\zug{S , c}, & \text{if $\sigma = \#$}\\
\zug{ \emptyset, a}, & \text{if $\sigma = \$$}
\end{cases}
\hspace{1cm}
\delta(\zug{S, c}, \sigma)=
\begin{cases}
\zug{\emptyset, c}, & \text{if $\sigma \neq \$$}\\
\zug{ \emptyset, a}, & \text{if $\sigma = \$$}
\end{cases}
$$

Also, $\alpha=\{\zug{\zug{S,c},\sigma,\zug{\emptyset,c}} : \sigma \in S\}$. 

Note that reading $\$$ from any state we move to $\zug{\emptyset, a}$ to detect a good infix,  and if at some point we read an unexpected letter, we move to $\zug{\emptyset, c}$, where we wait for $\$$ to be read. It is easy to see that $\D_n$ recognizes $L_n$ and has $2^{n+1}$ states. 

	\begin{figure}[htb]	
		\begin{center}
			\includegraphics[width=0.6\textwidth]{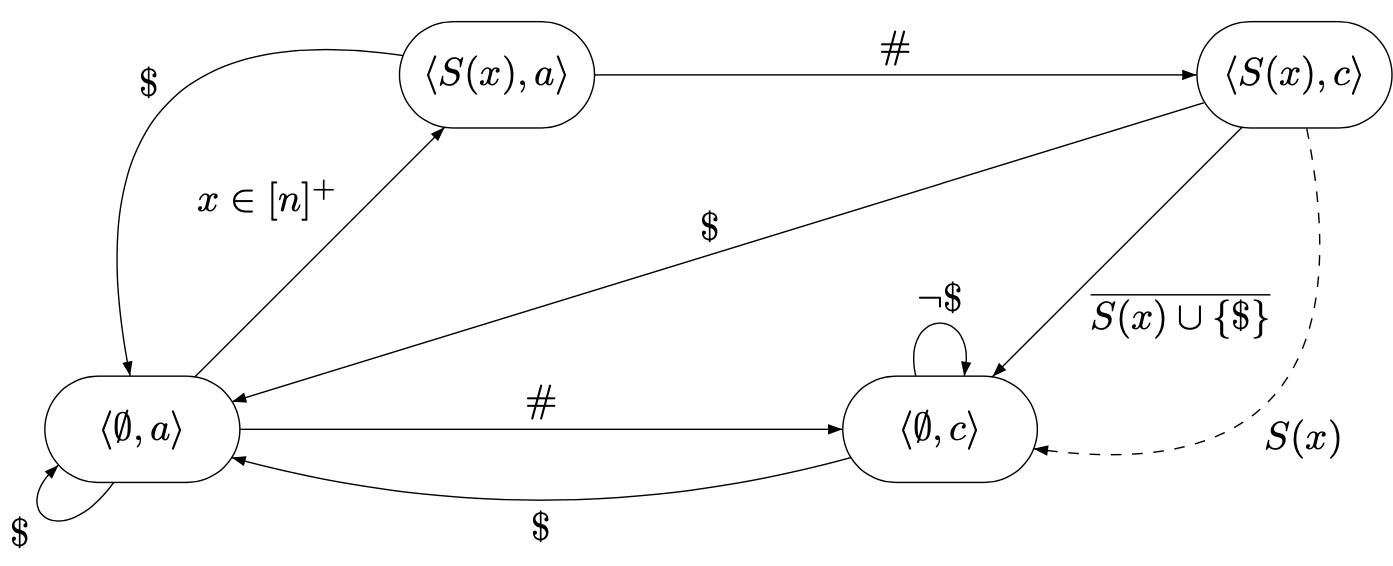}
			\caption{The tDBW  $D_n$. For a finite word $x\in [n]^+$, we denote the set of numbers that appear in $x$ by $S(x)$. 
				Dashed transitions are $\alpha$-transitions.}
			\label{Dn}
		\end{center}
	\end{figure}

\subsection{Missing details in Remark~\ref{alt min co proof remark}}
\label{app alt min co proof remark}

Consider the tDCW $\D_n  =\zug{\Sigma_n \cup \{\$\}, (2^{[n]}\times \{a, c\} )\cup \{q_{pass}\}, \zug{\emptyset, c}, \delta, \alpha}$ defined as follows (see Figure~\ref{DnC}).
The state space of $\D_n$  consists of a state $q_{pass}$ and two copies of subsets of $[n]$: the $a$-copy  $2^{[n]}\times \{a\}$, with ``a" standing for ``accumulate", and the $c$-copy $2^{[n]}\times \{c\}$, with ``c" standing for ``check".
Essentially, the state $\zug{S, a}$ in the $a$-copy remembers that we have read only numbers in $[n]$ after the last $\$$ and also remembers the set $S\subseteq [n]$ of these numbers.  The state $\zug{S, c}$ in the $c$-copy remembers, in addition, that we have just read $\#$. 
Accordingly, the state $\zug{S, c}$, for a nonempty set $S$,  checks whether the next letter is in $S$ and moves to the state $q_{pass}$  via a transition in $\overline{\alpha}$ only when the check passes. 
The fact that the above transitions are in $\overline{\alpha}$,  and we can move from $q_{pass}$ back to $\zug{\emptyset, a}$ upon reading $\$$ and without traversing $\alpha$, imply that $q_{pass}$ can trap runs that do not traverse $\alpha$ on words of the form $z_1\cdot z_2\cdots $, where for all $l\geq 1$ $z_l$ is good.
Formally,  for every state $\zug{S, o} \in  2^{[n]}\times \{a, c\}$, and letter $\sigma \in \Sigma_n\cup \{\$\}$, we define:

$$
\delta(\zug{S, o}, \sigma)=
\begin{cases}
\zug{S \cup \{\sigma\}, a},  &  \text{if $o=a$ and $\sigma \in [n]$} \\
\zug{S , c}, & \text{if $o = a, S\neq \emptyset$ and $\sigma = \#$}\\
q_{pass}, & \text{if $o = c$ and $\sigma \in S$}\\
\zug{\emptyset, c}, & \text{if $\zug{S, o} = \zug{\emptyset, c}$ and $\sigma = \neg \$$}\\	
\zug{\emptyset, c}, & \text{if $o = c$ and $\sigma \in ([n]\cup \{\#\})\setminus S$}\\
\zug{\emptyset, a}, & \text{if $\sigma=\$$}\\
\zug{ \emptyset,c}, & \text{if $\zug{S, o} = \zug{\emptyset, a}$ and $\sigma = \#$}

\end{cases}
$$

where all types transitions are in $\overline{\alpha}$ except for the last three. In addition, by reading a letter $\sigma$ from $q_{pass}$ we move to $\zug{\emptyset, a}$ via a transition in $\overline{\alpha}$ when $\sigma=\$$, and otherwise, we move to $\zug{\emptyset, c}$ via an $\alpha$-transition.
Note that by  reading $\$$ from any state we move to $\zug{\emptyset, a}$ to detect a suffix of the form $z_1\cdot z_2\cdots $, where for all $l\geq 1$ $z_l$ is good, and if at some point we read an unexpected letter, we move to $\zug{\emptyset, c}$, where we wait for $\$$ to be read. 

It is easy to see that $\D_n$ recognizes $L_n$ and has $2^{n+1} + 1$ states.

\begin{figure}[htb]	
\begin{center}
	\includegraphics[width=0.8\textwidth]{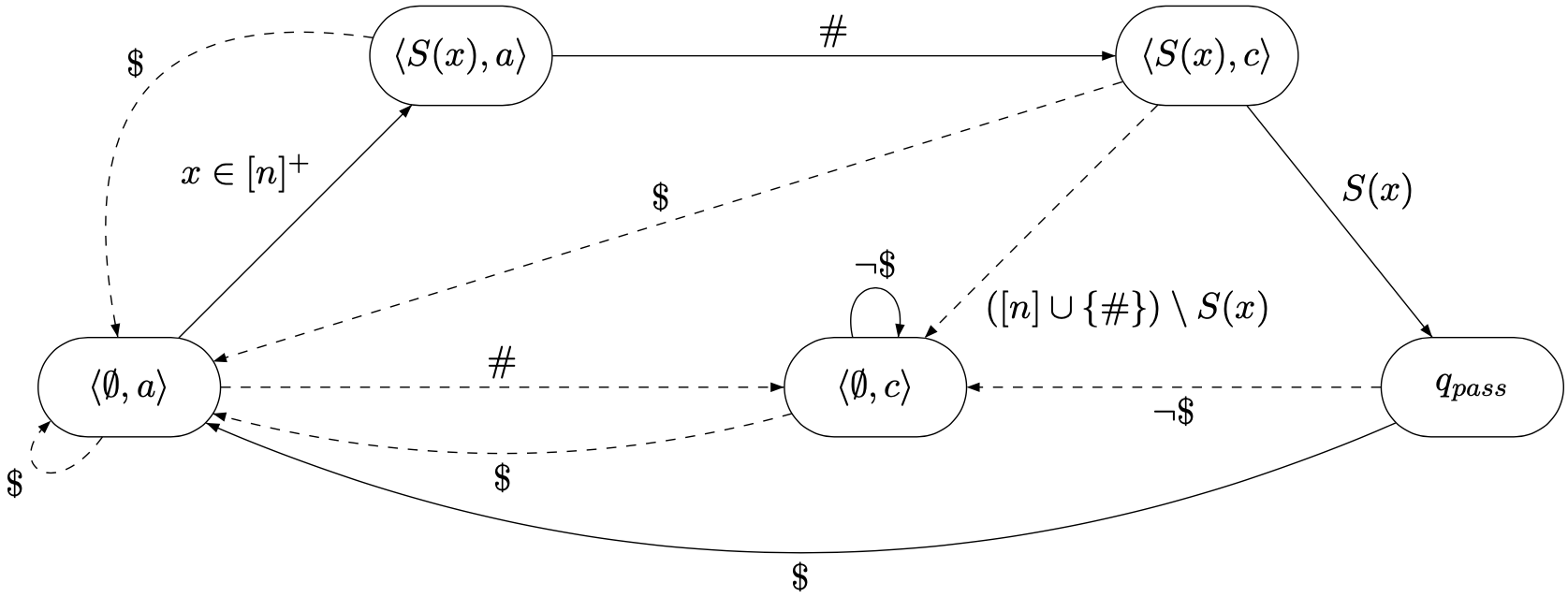}
	\caption{The tDCW $D_n$. For a finite word $x\in [n]^+$, we denote the set of numbers that appear in $x$ by $S(x)$. 
		Dashed transitions are $\alpha$-transitions.}
	\label{DnC}
\end{center}
\end{figure}

\end{document}